\documentclass[aps,twocolumn,groupedaddress]{revtex4}
\usepackage{amsmath,amsthm,verbatim} 
\usepackage{amssymb}
\usepackage{graphicx}
\usepackage{bbm}
\usepackage{subfigure}
\usepackage[rgb]{xcolor}
\usepackage{mathrsfs}
\usepackage{stackrel} 
\usepackage[all]{xy}
\usepackage{float} 

\newcommand{\beq}{\begin{equation}}
\newcommand{\eeq}{\end{equation}}

\newcommand{\C}{\mathbb{C}}
\newcommand{\N}{\mathbb{N}} 
\newcommand{\R}{\mathbb{R}} 
\renewcommand{\1}{\mathbbm{1}} 

\newcommand{\cH}{\mathcal{H}}
\newcommand{\cK}{\mathcal{K}}
\newcommand{\cE}{\mathcal{E}}
\newcommand{\cF}{\mathcal{F}}
\newcommand{\cD}{\mathcal{D}}
\newcommand{\cS}{\mathcal{S}}
\newcommand{\cR}{\mathcal{R}}
\newcommand{\cN}{\mathcal{N}}
\newcommand{\cC}{\mathcal{C}}
\newcommand{\cP}{\mathcal{P}}
\newcommand{\cT}{\mathcal{T}}
\newcommand{\cU}{\mathcal{U}}
\newcommand{\fS}{\mathfrak{S}}

\newcommand{\fR}{\mathfrak{R}}
\newcommand{\fI}{\mathfrak{I}}
\newcommand{\B}{\mathcal{B}}
\newcommand{\D}{\mathcal{D}}
\newcommand{\tr}{\mathrm{tr}}
\newcommand{\eps}{\varepsilon}

\definecolor{myurlcolor}{rgb}{0,0,0.4}
\definecolor{mycitecolor}{rgb}{0,0.5,0}
\definecolor{myrefcolor}{rgb}{0.5,0,0}
\usepackage{hyperref}
\hypersetup{colorlinks,
linkcolor=myrefcolor,
citecolor=mycitecolor,
urlcolor=myurlcolor}

\makeatletter
\newtheorem*{rep@theorem}{\rep@title}
\newcommand{\newreptheorem}[2]{%
\newenvironment{rep#1}[1]{%
 \def\rep@title{#2 \ref{##1}}%
 \begin{rep@theorem}}%
 {\end{rep@theorem}}}
\makeatother

\newtheorem{theorem}{Theorem} 
\newreptheorem{theorem}{Theorem}
\newtheorem{lemma}[theorem]{Lemma}
\newreptheorem{lemma}{Lemma}

\newtheorem{corol}[theorem]{Corollary}
\newtheorem{defn}[theorem]{Definition}

\usepackage{mathtools}
\usepackage[normalem]{ulem}

\def\d{\mathrm{d}}

\def\eps{\varepsilon}

\def\4s{\left| \hspace{-0.25em}\begin{smallmatrix} \ell,m \\ r,s\end{smallmatrix}\hspace{-0.25em}\right\rangle}

\begin{document}

\title{Security of continuous-variable quantum key distribution via a Gaussian de Finetti reduction}
\author{Anthony Leverrier}
\affiliation{Inria Paris, France.}
\email{anthony.leverrier@inria.fr}

\date{\today}

\begin{abstract}
Establishing the security of continuous-variable quantum key distribution against general attacks in a \textit{realistic} finite-size regime is an outstanding open problem in the field of theoretical quantum cryptography if we restrict our attention to protocols that rely on the exchange of coherent states. Indeed, techniques based on the uncertainty principle are not known to work for such protocols, and the usual tools based on de Finetti reductions only provide security for unrealistically large block lengths. 
We address this problem here by considering a new type of \textit{Gaussian} de Finetti reduction, that exploits the invariance of some continuous-variable protocols under the action of the unitary group $U(n)$ (instead of the symmetric group $S_n$ as in usual de Finetti theorems), and by introducing generalized $SU(2,2)$ coherent states. Crucially, combined with an energy test, this allows us to truncate the Hilbert space globally instead as at the single-mode level as in previous approaches that failed to provide security in realistic conditions. 
Our reduction shows that it is sufficient to prove the security of these protocols against \textit{Gaussian} collective attacks in order to obtain security against general attacks, thereby confirming rigorously the widely held belief that Gaussian attacks are indeed optimal against such protocols.
\end{abstract}

\maketitle

Quantum key distribution (QKD) is a cryptographic primitive aiming at distributing large secret keys to two distant parties, Alice and Bob, who have access to an authenticated classical channel. Mathematically, a QKD protocol $\cE$ is described by a quantum channel, that is a completely positive trace-preserving (CPTP) map transforming an input state, typically a large bipartite entangled state shared by Alice and Bob, into two keys, ideally two identical bit strings unknown to any third party. Establishing the security of the protocol against arbitrary attacks means proving that the map $\cE$ is approximately equal to an ideal protocol $\cF$.
An operational way of quantifying the security is by bounding the completely positive trace distance, or diamond distance between the two maps \cite{PR14}: the protocol is said to be $\eps$-secure if $\|\cE - \cF\|_{\diamond} \leq \eps$. 
If $\cE$ and $\cF$ act on some Hilbert space $\cH$ and $\Delta = \cE - \cF$, then the diamond norm is defined as
\begin{align}
\|\Delta\|_\diamond = \sup_{\rho \in \fS(\cH \otimes \cH') } \|(\Delta \otimes \1_{\cH'}) (\rho) \|_1 \label{eqn:diamond1}
\end{align}
where $\|\cdot\|_1$ is the trace norm and $\fS(\cH \otimes \cH')$ is the set of normalized density matrices (positive operators of trace 1) on $\cH \otimes \cH'$ with $\cH' \cong \cH$ (see \textit{e.g.} \cite{wat16}).
Computing an upper bound of Eq.~\eqref{eqn:diamond1} is very challenging in general because the Hilbert space $\cH = \cH_1^{\otimes n}$ has a dimension scaling exponentially with the number $n$ of quantum systems shared by Alice and Bob. Typical values of $n$ range in the millions or billions.  

In order to estimate the diamond norm, it is natural to exploit all the symmetries displayed by $\Delta$. For instance, if $\cE$ is a QKD protocol involving many 2-qubit pairs, such as BB84 for instance \cite{BB84}, then $\Delta$ might be covariant under any permutation of these pairs. For such maps, Christandl, K\"onig and Renner \cite{CKR09} showed that the optimization of Eq.~\eqref{eqn:diamond1} can be dramatically simplified provided that one is only interested in a polynomial approximation of $\|\Delta\|_\diamond$: indeed, it is then sufficient to consider a \textit{single} state, called a ``de Finetti state'', instead of optimizing over $\cH \otimes \cH' \cong (\C^4)^n \otimes (\C^4)^n$. More precisely, this de Finetti state is a purification of 
$ \tau_{\cH} = \int \sigma_{\cH_1}^{\otimes n} \mu(\sigma_{\cH_1})$, where $\cH_1 \cong \C^{4}$ is the single-system Hilbert space and $\mu(\cdot)$ is the measure on the space of density operators on a single system induced by the Hilbert-Schmidt metric.

This approach, called a \textit{de Finetti reduction}, has been applied successfully to analyze the security of QKD protocol such as BB84 \cite{SLS10} or qudit protocols \cite{SS10}. Indeed, computing the value of $\|(\Delta \otimes \1_{\cH'}) \tau_{\cH \cH'}\|_1$ for some purification $\tau_{\cH \cH'}$ of $\tau_{\cH}$ is usually tractable and is closely related to the task of establishing the security of the QKD protocol against \textit{collective attacks}, corresponding to restricting the inputs of $\cE$ to i.i.d. states of the form $\sigma_{\cH_1}^{\otimes n}$. 
A full security proof then consists of two steps: proving the security against these restricted collective attacks, and applying the de Finetti reduction to obtain security (with a polynomially larger security parameter) against general attacks. 

An outstanding problem in the theory of QKD is to address the security of protocols with continuous variables, that is protocols encoding the information in the continuous degrees of freedom of the quantified electro-magnetic field \cite{WPG12, DL15}. From a practical point of view, the essential difference between continuous-variable (CV) protocols and discrete-variables ones lies in the detection method: CV protocols rely on coherent detection, either homodyne or heterodyne depending on whether one or two quadratures are measured for each mode, while discrete-variable protocols use photon counting. 
The main theoretical difference is the Hilbert space $\cH$, which is \textit{infinite-dimensional} for CV QKD, corresponding to a $2n$-mode Fock space: $\cH = F(\C^n \otimes \C^n)=\bigoplus_{k=0}^\infty \mathrm{Sym}^k(\C^{n} \otimes \C^{n})$, 
where $\mathrm{Sym}^k(H)$ stands for the symmetric part of $H^{\otimes k}$. Note that the definition of Eq.~\eqref{eqn:diamond1} is formally restricted to finite-dimensional spaces, but we will ignore this issue here because one can always truncate $\cH$ to make its dimension finite (arbitrary large) and will therefore assume that the supremum can still be taken on $\cH\otimes \cH'$ for $\cH' \cong \cH$.  
For later convenience, let us denote $\cH$ by $F_{1,1,n}$ and $\cH \otimes \cH'$ by $F_{2,2,n} :=F(\C^{2n} \otimes \C^{2n}) \cong F_{1,1,n}\otimes F_{1,1,n}$.

A possible strategy to prove the security of such CV protocols is to follow the same steps as for BB84: first establish the security against collective attacks, then prove that this implies security against general attacks (with a reasonable loss). 
For protocols involving a Gaussian modulation of coherent states and heterodyne detection \cite{WLB04}, composable security against collective attacks was recently demonstrated in \cite{Lev15}. The second step is to apply the de Finetti reduction outlined above. The difficulty here comes from the infinite dimensionality of the Fock space $\cH$. In order to apply the technique of \cite{CKR09}, it is therefore needed to truncate the Fock space in a suitable manner. This can be achieved with the help of an energy test, but unfortunately, the local dimension of $\overline{\cH}_1$, the truncated single-mode space, needs to grow like the logarithm of $n$, for the technique to apply \cite{LGRC13}. Indeed, the technique of \cite{CKR09} was developed for finite-dimensional systems, and the energy test needs to enforce that with high probability, \emph{each the unmeasured modes} contains a number of photons below some given threshold. Such a guarantee can only be obtained for a threshold increasing logarithmically with $n$. The dimension of the total truncated Hilbert space is then super-exponential in $n$, on the order of $(\log n)^{Cn}$, for some constant $C>1$. Since the loss in the security parameter obtained with \cite{CKR09} is superpolynomial in the dimension of the total Hilbert space, this means that if the protocol is $\eps$-secure against collective attacks, this approach only shows that the protocol is also $\eps'$-secure against general attacks with $\eps' = \eps \times 2^{\mathrm{polylog}(n)}$. While this gives a proof that the protocol is asymptotically secure in the limit of infinitely large block lengths, it fails to provide any useful bound in practical regimes where $n \sim 10^6 - 10^9$. We note that a related strategy relies on the exponential de Finetti theorem but fails similarly to provide practical security bounds in the finite-size regime \cite{Ren08,RC09}.

Let us also mention that there exists a CV QKD protocol with proven security where Alice sends squeezed states to Bob instead of coherent states \cite{CLV01}. This protocol can be analyzed thanks to an entropic uncertainty relation \cite{FBB12}, but this technique requires the exchange of squeezed states, which makes the protocol much less practical. Moreover, this approach does not recover the secret key rate corresponding to Gaussian attacks in the asymptotic limit of large $n$, even though these attacks are expected to be optimal. Here, in contrast, we are interested in the security of CV protocols based on the exchange of coherent states. 

The idea that we exploit in this paper is that CV QKD protocols not only display the permutation invariance common to most QKD protocols, but also a specific symmetry with a continuous-variable flavor \cite{LKG09}. 
This new symmetry is linked to the unitary group $U(n)$ instead of the symmetric group $S_n$. More precisely, the protocols are covariant if Alice and Bob process their $n$ respective modes with linear-optical networks acting like the unitary $u \in U(n)$ on Alice's annihilation operators and its complex conjugate $\overline{u}$ on Bob's annihilation operators. 

Our main technical result is an upper bound on $\|\Delta\|_{\diamond}$ for maps $\Delta$ covariant under a specific representation of the unitary group. For such maps, we show that is it sufficient to consider again a single state, which is the purification of a specific mixture of \textit{Gaussian} i.i.d.~states. This in turn will imply that it is sufficient to establish the security of the protocol against \textit{Gaussian} collective attacks in order to prove the security of the protocol against general attacks. 
An important technicality is that we still need to truncate the total Hilbert space to replace it by a finite-dimensional one. Crucially, this truncation can now be done globally and not for single-mode Fock spaces as in \cite{LGRC13} and this is this very point that makes our approach so effective.
Indeed, in our security proof, we argue that it is sufficient to consider states that are invariant under the action of $U(n)$ and such states live in a very small subspace of the ambient Fock space. More precisely, the dimension of the restriction of this subspace to states containing $K$ photons grows polynomially in $K$, instead of exponentially in the case of the total Fock space. This phenomenon is reminiscent of the fact that the dimension of the symmetric subspace of $(\C^{\otimes d})^{\otimes n}$ only grows polynomially in $n$ if the local dimension $d$ is constant.

The consequence is that the security loss due to the reduction from general to collective attacks will not scale like $2^{\mathrm{polylog}(n))}$ anymore, but rather like $O(n^4)$, which behaves \emph{much more nicely} for typical values of $n$, and  yields the first practical security proof of a CV QKD protocol with coherent states against general attacks. Indeed, our security reduction performs even better than the original de Finetti reduction developed for BB84, where the security loss scales like $O(n^{15})$ \cite{CKR09}.

Ideally, truncating the Fock space could be done by projecting the quantum state given as an input to $\Delta$ onto a finite dimensional subspace with say, less than $K$ photons (where the value of $K$ scales linearly with the total number of modes).
Of course, such a projection $\cP$ is unrealistic, and one will instead apply an energy test $\cT$ that passes if the energy measured on a small number $k \ll n$ of modes is below some threshold and will abort the protocol otherwise. Such an idea was already considered in previous works dealing with the security of CV QKD \cite{RC09, LGRC13, fur14}.
An application of the triangle inequality (see Lemma \ref{lem:sec-red} in the appendix) yields:
\begin{align}
\|\Delta \circ \cT \|_\diamond \leq \|\Delta \circ \cP\|_\diamond + 2 \|(\1-\cP)\circ \cT\|_\diamond. \label{eqn:triangle}
\end{align}
In other words, it is sufficient for our purposes to show the security of the protocol restricted to input states subject to a maximum photon number constraint, provided that we can bound the value of $\|(\1-\cP)\circ \cT\|_\diamond$, which corresponds to the probability that the energy test passes but that the number of photons in the remaining modes is large.

{\bf Analysis of the energy test}.---We show that $\|(\1-\cP)\circ \cT\|_\diamond$ is indeed small for a maximal number of photons $K$ scaling linearly with $n$ (see Appendix \ref{sec:test}). 
The energy test $\cT(k, d_A, d_B)$ depends on 3 parameters: the number $k$ of additional modes that will be measured for the test and maximum allowed average energies $d_A$ and $d_B$ for Alice and Bob's modes.  The input of the state is a $2(n+k)$-mode state. Alice and Bob should symmetrize this state by processing them with random conjugate linear-optical networks and measure the last $k$ modes with heterodyne detection, corresponding to a projection of standard (Glauber) coherent states. If the average energy per mode is below $d_A$ for Alice and $d_B$ for Bob, the test passes and Alice and Bob apply the protocol $\cE_0$ to their remaining modes. Otherwise the protocol simply aborts. These thresholds $d_A$, $d_B$ should be chosen large enough to ensure that the energy test passes with large probability. 
 Note that the symmetrization of the state can be done on the classical data for the protocols of Refs \cite{WLB04,POS15} since these protocols require both parties to measure all the modes with heterodyne detection, which itself commutes with the action of the linear-optical networks. For this, Alice and Bob need to multiply their measurement results (gathered as vectors for $\mathbbm{R}^{2n}$) by an identical random orthogonal matrix. There is also hope that this symmetrization can be further simplified, but we do not address this issue here.

{\bf An upper bound on $\|\Delta \circ \cP\|_\diamond$ via de Finetti reduction}.--- This requires two main ingredients: first, a proof that any mixed state on $F_{1,1,n}$ that is invariant under the action of the unitary group admits a purification in the \textit{symmetric subspace} $F_{2,2,n}^{U(n)}$, and second, that Gaussian states resolve the identity on the symmetric subspace. 
The symmetric subspace $F_{2,2,n}^{U(n)}$ was introduced and studied in Ref.~\cite{lev16} and is defined as follows:
\begin{align*}
F_{2,2,n}^{U(n)} = \left\{|\psi\rangle \in F_{2,2,n} \: :\: W_u |\psi\rangle = |\psi\rangle, \forall u \in U(n) \right\},
\end{align*}
where $u \mapsto W_u$ is a representation of the unitary group $U(n)$ on the Fock space $F_{2,2,n}$ corresponding to mapping the $4n$ annihilation operators $\vec{a} = (a_1, \ldots, a_n), \vec{b} =(b_1, \ldots, b_n), \vec{a}' = (a'_1, \ldots, a'_n), \vec{b}'=(b'_1, \ldots, b'_n)$ of each of the $n$ modes of $\cH_A, \cH_B, \cH_{A'}, \cH_{B'}$ to $u \vec{a}, \overline{u} \vec{b}, \overline{u} \vec{a}', u \vec{b}'$. Here $\overline{u}$ denotes the complex conjugate of $u$ and $F_{2,2,n} = \cH_A\otimes \cH_B \otimes \cH_{A'}\otimes \cH_{B'}$.

In Ref.~\cite{lev16}, a full characterization of the symmetric subspace $F_{2,2,n}^{U(n)}$ is given. Let us introduce the four operators $Z_{11}, Z_{12}, Z_{21}, Z_{22}$ defined by:
\begin{align*}
Z_{11} &= \sum_{i=1}^n a_i^\dagger b_i^\dagger, \quad Z_{12} =\sum_{i=1}^n a_i^\dagger a'^\dagger_i,\\
Z_{21} &= \sum_{i=1}^n b_i^\dagger b'^\dagger_i, \quad Z_{22} =\sum_{i=1}^n a'^\dagger_i b'^\dagger_i.
\end{align*}
We now define the so-called $SU(2,2)$ \textit{generalized coherent states} \cite{per72,per86}: to any $2\times 2$ complex matrix $\Lambda = \left(\begin{smallmatrix} \lambda_{11} & \lambda_{12} \\ \lambda_{21} & \lambda_{22} \end{smallmatrix} \right)$ such that $\Lambda \Lambda^\dagger \prec \1_2$ (that is, with a spectral norm strictly less than 1), we associate the $4n$-mode Gaussian state $|\Lambda, n\rangle = |\Lambda,1\rangle^{\otimes n}$ given by
\begin{align*}
|\Lambda,n\rangle = \mathrm{det} (1-\Lambda\Lambda^\dagger)^{n/2} \exp\left( \sum_{i,j=1}^2 \lambda_{ij} Z_{ij}\right) |\mathrm{vacuum}\rangle.
\end{align*}
Since the polynomial $\sum_{i,j=1}^2 \lambda_{ij} Z_{ij}$ is quadratic in the creation operators, the generalized coherent state is a Gaussian state. More specifically, it corresponds to $n$ copies of a centered 4-mode pure Gaussian state whose covariance matrix is a function of $\Lambda$ (see the discussion in Section 3 of Ref.~\cite{lev16} for details).

These generalized coherent states span the symmetric subspace \cite{lev16}, and moreover, for $n\geq 4$, they resolve the identity on the symmetric subspace \cite{lev16}:
\begin{align}
\int_{\cD} |\Lambda,n\rangle \langle \Lambda,n| \d \mu_n(\Lambda) = \1_{F_{2,2,n}^{U(n)}} \label{eqn:maintext-id}
\end{align}
where $\cD$ is the set of $2\times 2$ matrices $\Lambda$ such that $\Lambda\Lambda^\dagger \prec \1_2$ and $\mathrm{d}\mu_n(\Lambda)$ is the invariant measure on $\D$ given by
\begin{align*}
\mathrm{d}\mu_n (\Lambda) =  \frac{(n-1)(n-2)^2(n-3)}{\pi^{4}\det(\mathbbm{1}_2 - \Lambda \Lambda^\dagger)^4 } \d \lambda_{11} \d \lambda_{12}\d \lambda_{21} \d \lambda_{22}. 
\end{align*}

Since the space $F_{2,2,n}^{U(n)}$ is infinite-dimensional, the integral of Eq.~\eqref{eqn:maintext-id} is not normalizable. In order to obtain an operator with finite norm, we consider the finite-dimensional subspace $F_{2,2,n}^{U(n), \leq K}$ of $F_{2,2,n}^{U(n)}$ spanned by states with less than $K$ ``excitations'':
\begin{align*}
\mathrm{Span}\left\{ (Z_{11})^i (Z_{12})^j (Z_{21})^k (Z_{22})^\ell |\mathrm{vac}\rangle \: : \: i+j+k+\ell \leq K \right\}.
\end{align*}
We show in Appendix \ref{sec:finite} that an approximate resolution of the identity still holds for this space when restricting the coherent states $|\Lambda,n\rangle$ to $\Lambda \in \cD_\eta$ for $\cD_\eta = \left\{ \Lambda \in \cD \: : \: \eta \1_2 -\Lambda\Lambda^\dagger \succeq 0Ê\right\}$ for $\eta \in [0,1[$. Let us denote by $\Pi_{\leq K}$ the identity onto the subspace $F_{2,2,n}^{U(n), \leq K}$ and introduce the relative entropy $D(x||y) = x \log \frac{x}{y} + (1-x) \log \frac{1-x}{1-y}$. 
\begin{theorem} \label{thm:finite-version}
For $n\geq 5$ and $\eta \in [0,1[$, if $K \leq\frac{\eta N}{1-\eta}$ for $N=n-5$, then the operator inequality 
\begin{align}
\int_{\mathcal{D}_\eta} |\Lambda,n\rangle \langle \Lambda,n | \d \mu_n(\Lambda)\geq (1-\eps) \Pi_{\leq K} \label{eqn:approximate-main}
\end{align}
holds with $\eps = 2 N^4 (1+K/N)^7 \exp\left(-N D\left(\frac{K}{K+N} \big\| \eta \right) \right)$.
\end{theorem}

This approximate resolution of the identity allows us to bound the diamond norm of maps which are covariant under the action $W_u$ of the unitary group $U(n)$, provided that the total photon number of the input state is upper bounded by some known value $K$. Let us define $\tau_{\cH}^\eta$ to be the normalized state corresponding to the left-hand side of Eq.~\eqref{eqn:approximate-main}, and $\tau_{\cH \cN}^{\eta}$ a purification of $\tau_{\cH}^\eta$. 
\begin{theorem}\label{thm:postselection}
Let $\Delta: \mathrm{End}(F_{1,1,n}^{\leq K}) \to  \mathrm{End}(\mathcal{H}')$ such that for all $u \in U(n)$, there exists a CPTP map $\mathcal{K}_u: \mathrm{End}(\mathcal{H}') \to \mathrm{End}(\mathcal{H}')$ such that $\Delta \circ W_u = \mathcal{K}_u \circ \Delta$, then 
\begin{align*}
\|\Delta\|_\diamond \leq \frac{K^4}{50} \|(\Delta \otimes \mathrm{id}) \tau^\eta_{\mathcal{H}\mathcal{N}} \|_1,
\end{align*}
for $\eta = \frac{K-n+5}{K+n-5}$, provided that $n \geq N^*(K/(n-5))$.
\end{theorem}
The function $N^*$ is defined in Eq.~\eqref{eqn:N*} and its argument is an upper bound on the average number of photons per mode. One has for instance $N^*(21) \approx 10^4$, $N^*(60) \approx 10^5$. 
 
Similarly as in \cite{CKR09} for the case of permutation invariance, Theorem \ref{thm:postselection} shows that one can obtain a polynomial approximation of degree 4 (if the average number of photons per mode is constant) of the diamond norm by simply evaluating the trace norm of the map on a very simple state, namely a purification of a mixture of Gaussian i.i.d. states. 
We note that we restricted the analysis to $SU(2,2)$ coherent states here because they are the relevant ones for cryptographic applications, but our results can be extended to $SU(p,q)$ coherent states for arbitrary integers $p, q$. In that case, the prefactor of the diamond norm approximation would be a polynomial of degree $pq$.

{\bf Security reduction to Gaussian collective attacks}.---We now explain how to obtain a bound on  $\|(\Delta \otimes \mathrm{id}) \tau^\eta_{\mathcal{H}\mathcal{N}} \|_1$, if we already know that the initial protocol (without the energy test) is $\eps$-secure against collective attacks. Let us therefore assume that we are given such a CV QKD protocol $\cE_0$ acting on $2n$-mode states shared by Alice and Bob which is, in addition, covariant under the action of the unitary group (\textit{i.e.} there exists $\mathcal{K}_u$ such that $\cE_0 \circ W_u = \cK_u \circ \cE_0$). Examples of such protocols are the no-switching protocol \cite{WLB04} and the measurement-device-independent protocol of Ref.~\cite{POS15}, provided that they are suitably symmetrized.
We define $\cE := \cR \circ \cE_0 \circ \cT$, where $\cR$ is an additional privacy amplification step that reduces the key by $\lceil 2 \log_2 \tbinom{K+4}{4} \rceil$ bits.

Recall that by definition, the QKD protocol $\cE_0$ is $\eps$-secure against Gaussian collective attacks if
\begin{align*}
\| ((\cE_0-\cF_0)\otimes \mathrm{id})(|\Lambda,n\rangle \langle \Lambda,n|)\|_1\leq \eps 
\end{align*}
for all $\Lambda \in \cD$. It means that the protocol is shown to be secure for input states of the form $\tr_{\cH_{A'}\cH_{B'}} (|\Lambda, n\rangle\langle \Lambda,n|)$, which are nothing but i.i.d. bipartite Gaussian states. By linearity, we immediately obtain that $\|(\Delta \otimes \mathrm{id}) \tau^\eta_{\mathcal{H}} \|_1 \leq \eps$. To finish the proof, we need to take into account the extra system $\cN$ given to Eve. This system can be chosen of dimension $\tbinom{K+4}{4}$ and the leftover hashing lemma of Renner \cite{Ren08} says that by shortening the final key of the protocol by $2 \log_2 (\mathrm{dim} \, \cN)$, one ensures that the protocol remains $\eps$-secure. This is the role of the map $\cR$. Overall, we find that $\|((\cE-\cF) \otimes \mathrm{id}) \tau^\eta_{\mathcal{H}\mathcal{N}} \|_1 \leq \eps$.

{\bf Results}.---Putting everything together, we show that if $\cE_0$ is covariant under the action of the unitary group and $\eps$-secure against Gaussian collective attacks, then the protocol $\cE = \cR \circ \cE_0 \circ \cT$ is $\eps'$-secure against general attacks, with
\begin{align}
\eps' =\frac{K^4}{50} \eps \label{eqn:final-result-main}
\end{align}
for $K = \max \Big\{1, n(d_A + d_B)\Big(1 + 2 \sqrt{\frac{\ln (8/\eps)}{2n}} +  \frac{\ln (8/\eps)}{n}\Big)\Big(1-2{\sqrt{\frac{\ln (8/\eps)}{2k}}}\Big)^{-1}\Big\}$. The full proof is presented in Appendix \ref{sec:general-proof}. 
The advantage of our approach compared to the previous results of \cite{LGRC13} is two-fold: first the improvement of the prefactor in Eq.~\eqref{eqn:final-result-main} from $2^{\mathrm{polylog}(n)}$ to $O(n^4)$ yields security for practical settings; second, it is only required to establish the security of the protocol against Gaussian collective attacks in order to apply our security reduction, a task arguably much simpler than addressing the security against collective attacks in the case of CV QKD.

{\bf Discussion}.---Despite their wide range of application, there is a regime where ``standard'' de Finetti theorems fail, namely when the local dimension is not negligible compared to the number $n$ of subsystems \cite{CKMR07}. In particular, these techniques do not apply directly to CV protocols where the local spaces are infinite-dimensional Fock spaces. 
In this work, we considered a natural symmetry displayed by some important CV QKD protocols, which are covariant under the action of beamsplitters and phase-shifts on their $n$ modes \cite{LKG09}. For such protocols, one legitimately expects that stronger versions of de Finetti theorems should hold. In particular, a widely held belief that it is enough to consider \emph{Gaussian} i.i.d.~input states instead of all i.i.d.~states in order to analyze the security of the corresponding protocol. 

We proved this statement rigorously here. Our main tool is a family of $SU(2,2)$ generalized coherent states that resolve the identity of the subspace spanned by states invariant under the action of $U(n)$. This implies that in some applications such as QKD, it is sufficient to consider the behaviour of the protocol on these states in order to obtain guarantees that hold for arbitrary input states.

Let us conclude by discussing the issue of active symmetrization. For the proof above to go through, it is required that the protocols are covariant under the action of the unitary group. Such an invariance can be enforced by symmetrizing the classical data held by Alice and Bob. However, this step is computationally costly and it would be beneficial to bypass it. We believe that this should be possible. Indeed, it is often argued that a similar step is unnecessary when proving the security of BB84 for instance, and there is no fundamental reason to think that the situation is different here. Moreover, we already know of security proofs based on the uncertainty principle \cite{TR11, TLG12,DFR16} where such a symmetrization is not required.


\begin{acknowledgements}
I gladly acknowledge inspiring discussions with Matthias Christandl and Tobias Fritz.
\end{acknowledgements}

\appendix

\clearpage

\begin{widetext}
\begin{center} \textbf{\Large{Appendix}}
\end{center}

In Section \ref{sec:symm-sub}, we recall the main results from Ref.~\cite{lev16} about the symmetric subspace $F_{2,2,n}^{U(n)}$ and the generalized $SU(2,2)$ coherent states.
In Section \ref{sec:lemmas}, we present a series of technical lemmas and prove in Section \ref{sec:finite} that bounded-energy generalized coherent states approximately resolve the identity on $F_{2,2,n}^{U(n), \leq K}$. 
In Section \ref{sec:general-proof}, we explain how to perform the security proof of the protocol and show that bounding the norm of $\Delta = \cE - \cF$ decomposes into separate tasks.
In Section \ref{sec:generalization}, we derive our generalization of the de Finetti reduction of \cite{CKR09} to maps that are covariant under the action of the unitary group $U(n)$.
In Section \ref{sec:collective}, we show how to reduce the security analysis against general attacks to a security analysis against Gaussian collective attacks, if the photon number of the input states is bounded. 
Finally, in Section \ref{sec:test}, we analyze the energy test and show how it provides the restriction on the input states required for the proof of Section \ref{sec:collective} to go through.


\section{The symmetric subspace $F_{2,2,n}^{U(n)}$ and generalized $SU(2,2)$ coherent states}
\label{sec:symm-sub}

In this section, we recall some results from Ref.~\cite{lev16} where the symmetric subspace $F_{p,q,n}^{U(n)}$ is considered for arbitrary integers $p,q$ and specialize them to the case where $p=q=2$.

\subsection{The symmetric subspace $F_{2,2,n}^{U(n)}$}
Let  $H_A \cong H_B \cong H_{A'} \cong H_{B'} \cong \C^n$ and define the Fock space $F_{2,2,n}$ as   
\begin{align*}
F_{2,2,n} := \bigoplus_{k=0}^\infty \mathrm{Sym}^k(H_A \otimes H_B \otimes H_{A'} \otimes H_{B'}),
\end{align*} 
where $\mathrm{Sym}^k(H)$ is the symmetric part of $H^{\otimes k}$.

In this paper, we will use both the standard Hilbert representation and the Segal-Bargmann representation of $F_{2,2,n}$. 
Using the Segal-Bargmann representation, the Hilbert space $F_{2,2,n}$ is realized as a functional space of complex holomorphic functions square-integrable with respect to a Gaussian measure, $F_{2,2,n} \cong L^2_{\mathrm{hol}}(\C^{4n}, \| \cdot\|)$, with a state $\psi \in F_{2,2,n}$ represented by a holomorphic function $\psi(z,z')$ with $z \in \C^{2n}, z' \in \C^{2n}$ satisfying 
\begin{align} \label{eqn:norm}
\|\psi\|^2 := \langle \psi, \psi\rangle = \frac{1}{\pi^{4n}}\int \exp(-|z|^2 -|z'|^2) |\psi(z,z')|^2 \d z \d z'< \infty
\end{align}
where $\d z :=  \prod_{k=1}^n \prod_{i=1}^2 \mathrm{d}z_{k,i}$ and $\d z' := \prod_{k=1}^n \prod_{j=1}^2 \mathrm{d}z_{k,j}'$ denote the Lebesgue measures on $\C^{2n}$ and $\C^{2n}$, respectively,  and $|z|^2 := \sum_{k=1}^n\sum_{i=1}^2 |z_{k,i}|^2, |z'|^2 := \sum_{k=1}^n \sum_{j=1}^2 |z_{k,j}'|^2$.
A state $\psi$ is therefore described as a holomorphic function of $4n$ complex variables $(z_{1,1}, z_{n,1}; z_{1,2}, \ldots, z_{n,2};  z_{1,1}', \ldots, z_{n,1}'; z_{1,2}', \ldots, z_{n,2}')$. In the following, we denote by $z_i$ and $z_j'$ the vectors $(z_{1,i}, \ldots, z_{n,i})$ and $(z_{1,j}', \ldots, z_{n,j}')$, respectively, for $i,j \in \{1,2\}$. With these notations, the vector $z_1$ is associated to the space $H_A$, the vector $z_1'$ to $H_B$, the vector $z_2$ to $H_B'$ and the vector $z_2'$ to $H_A'$. These notations are chosen so that the unitary $u \in U(n)$ acts as $u$ on $z_1, z_2$, and $\overline{u}$ on $z'_1, z_2'$.

Let $\mathfrak{B}(F_{2,2,n})$ denote the set of bounded linear operators from $F_{2,2,n}$ to itself and let $\mathfrak{S}(F_{2,2,n})$ be the set of quantum states on $F_{2,2,n}$: positive semi-definite operators with unit trace.

Formally, one can switch from the Segal-Bargmann representation to the representation in terms of annihihation and creation operators by replacing the variables $z_{k,1}$ by $a_k^\dagger$, $z_{k,2}$ by $b'^\dagger_k$, $z'_{k,1}$ by $b_k^\dagger$ and $z'_{k,2}$ by $a'^\dagger_k$.
The function $f(z,z')$ is therefore replaced by an operator $f(a^\dagger, b^\dagger, a'^\dagger, b'^\dagger)$ and the corresponding state in the Fock basis is obtained by applying this operator to the vacuum state.

The metaplectic representation of the unitary group $U(n) \subset Sp(2n,\R)$ on $ F_{2,2,n}$ associates to $u \in U(n)$ the operator $W_u$ performing the change of variables $z \to uz$, $z' \to \overline{u} z'$:
\begin{align*}
U(n) & \to \mathfrak{B}(F_{2,2,n})\\
u & \mapsto W_u = \big[ \psi(z_1, z_2, z_1', z_2') \mapsto \psi(u z_1, u z_2, \overline{u} z_1',  \overline{u} z_2')\big]
\end{align*}
where $\overline{u}$ denotes the complex conjugate of the unitary matrix $u$.
In other words, the unitary $u$ is applied to the modes of $F_A \otimes F_{B'}$ and its complex conjugate is applied to those of $F_B \otimes F_{A'}$. 
 
The states that are left invariant under the action of the unitary group $U(n)$ are relevant for instance in the context of continuous-variable quantum key distribution, and we define the symmetric subspace as the space spanned by such invariant states. 
\begin{defn}[Symmetric subspace]
For integer $n \geq 1$, the \emph{symmetric subspace} $F_{2,2,n}^{U(n)}$ is the subspace of functions $\psi \in F_{2,2,n}$ such that
\begin{align*}
W_u \psi = \psi \quad \forall u \in U(n).
\end{align*} 
\end{defn}
The name \emph{symmetric subspace} is inspired by the name given to the subspace $\mathrm{Sym}^n(\mathbbm{C}^d)$ of $(\mathbbm{C}^d)^{\otimes n}$ of states invariant under permutation of the subsystems:
\begin{align}
\mathrm{Sym}^n(\mathbbm{C}^d) := \left\{|\psi\rangle \in(\mathbbm{C}^d)^{\otimes n} \: : \: P(\pi) |\psi\rangle = |\psi\rangle, \forall \pi \in S_n \right\}
\end{align}
where $\pi \mapsto P(\pi)$ is a representation of the permutation group $S_n$ on $(\mathbbm{C}^d)^{\otimes n}$ and $P(\pi)$ is the operator that permutes the $n$ factors of the state according to $\pi \in S_n$. See for instance \cite{har13} for a recent exposition of the symmetric subspace from a quantum information perspective.

In \cite{lev16}, a full characterization of the symmetric subspace $F_{2,2,n}^{U(n)}$ is given. It is helpful to introduce the four operators $Z_{11}, Z_{12}, Z_{21}, Z_{22}$ defined by:
\begin{align*}
Z_{11} =  \sum_{i=1}^n z_{i,1} z'_{i,1} \quad  & \leftrightarrow  \quad  \sum_{i=1}^n a_i^\dagger b_i^\dagger\\
Z_{12} =\sum_{i=1}^n z_{i,1} z'_{i,2} \quad  & \leftrightarrow  \quad   \sum_{i=1}^n a_i^\dagger a'^\dagger_i,\\
Z_{21} = \sum_{i=1}^n z_{i,2} z'_{i,1}  \quad  & \leftrightarrow  \quad   \sum_{i=1}^n b_i^\dagger b'^\dagger_i, \quad \\
Z_{22} =\sum_{i=1}^n z_{i,2} z'_{i,2}  \quad  & \leftrightarrow  \quad   \sum_{i=1}^n a'^\dagger_i b'^\dagger_i.
\end{align*}

\begin{defn}
For integer $n\geq 1$, let $E_{2,2,n}$ be the space of analytic functions $\psi$ of the $4$ variables $Z_{1,1}, \ldots, Z_{2,2}$, satisfying $\|\psi\|_E^2 < \infty$, that is $E_{2,2,n} = L^2_{\mathrm{hol}}(\C^{pq}, \|\cdot\|_E)$.
\end{defn}
In \cite{lev16}, is was proven that $E_{2,2,n}$ coincides with the symmetric subspace $F_{2,2,n}^{U(n)}$.

\begin{theorem}\label{thm:charact-symm}
For $n \geq 2$, the symmetric subspace $F_{2,2,n}^{U(n)}$ is isomorphic to $E_{2,2,n}$.
\end{theorem}

In other words, any state in the symmetric subspace can be written as
\begin{align*}
|\psi\rangle = f\big(\sum_{i=1}^n a_i^\dagger b_i^\dagger,  \sum_{i=1}^n a_i^\dagger a'^\dagger_i, \sum_{i=1}^n b_i^\dagger b'^\dagger_i, \sum_{i=1}^n a'^\dagger_i b'^\dagger_i \big) |\mathrm{vacuum}\rangle
\end{align*}
for some function $f$. 
Said otherwise, such a state is characterized by only 4 parameters instead of $4n$ for an arbitrary state in $F_{2,2,n}$; or else, the symmetric subspace is isomorphic to a 4-mode Fock space (with ``creation'' operators corresponding to $Z_{11}, Z_{12}, Z_{21}, Z_{22}$, instead of the ambient $4n$-mode Fock space.


\subsection{Coherent states for $SU(2,2)/SU(2)\times SU(2) \times U(1)$}
\label{sec:CS}

In this section, we first review a construction due to Perelomov that associates a family of generalized coherent states to general Lie groups \cite{per72}, \cite{per86}. In this language, the standard Glauber coherent states are associated with the Heisenberg-Weyl group, while the atomic spin coherent states are associated with $SU(2)$. The symmetric subspace $F_{2,2,n}^{U(n)}$ is spanned by $SU(2,2)$ coherent states, where $SU(2,2)$ is the special unitary group of signature $(2,2)$ over $\C$:
\begin{align}
SU(2,2) := \left\{ A \in M_{4}(\C) \: : \: A \1_{2,2} A^\dagger =\1_{2,2}Ê\right\}
\end{align}
where $M_{4}(\C)$ is the set of $4\times 4$-complex matrices and $\1_{2,2} = \1_{2} \oplus (-\1_2)$.

In Perelomov's construction, a \emph{system of coherent states of type} $(T, |\psi_0\rangle)$ where $T$ is the representation of some group $G$ acting on some Hilbert space $\mathcal{H} \ni |\psi_0\rangle$, is the set of states $\left\{|\psi_g\rangle \: : \: |\psi_g\rangle = T_g |\psi_0\rangle\right\}$ where $g$ runs over all the group $G$. One defines $H$, the \emph{stationary subgroup} of $|\psi_0\rangle$ as 
\begin{align*}
H := \left\{g \in G \: :Ê\: T_g |\psi_0\rangle = \alpha |\psi_0\rangle \, \text{for} \, |\alpha|=1Ê\right\},
\end{align*}
that is the group of $h \in G$ such that $|\psi_h\rangle $ and $|\psi_0\rangle$ differ only by a phase factor. When $G$ is a connected noncompact simple Lie group, $H$ is the maximal subgroup of $G$. 
In particular, for $G = SU(2,2)$, one has $H= SU(2,2) \cap U(4) = SU(2) \times SU(2)\times U(1)$ and the factor space $G/H$ corresponds to a Hermitian symmetric space of classical type (see \textit{e.g.} Chapter X of \cite{hel79}).  
The generalized coherent states are parameterized by points in $G/H$. For $G/H = SU(2,2)/SU(2)\times SU(2) \times U(1)$, the factor space is the set $\D$ of $2\times 2$ matrices $\Lambda$ such that $\Lambda \Lambda^\dagger < \1_{p}$, i.e.~the singular values of $\Lambda$ are strictly less than 1.
\begin{align*}
\D = \left\{ \Lambda \in M_{2}(\C) \: : \:\mathbbm{1}_2 - \Lambda\Lambda^\dagger >0 \right\},
\end{align*}
where $A>0$ for a Hermitian matrix $A$ means that $A$ is positive definite.

We are now ready to define our coherent states for the noncompact Lie group $SU(2,2)$. 
\begin{defn}[$SU(2,2)$ coherent states] \label{defn:CS} For $n \geq 1$, the coherent state $\psi_{\Lambda,n}$ associated with $\Lambda \in \D$ is given by
 \begin{align*}
 \psi_{\Lambda,n}(Z_{1,1}, \ldots, Z_{2,2}) = \det (1-\Lambda \Lambda^\dagger)^{n/2} \det \exp (\Lambda^T Z)
 \end{align*}
 where $Z$ is the $2\times 2$ matrix $\left[ Z_{i,j}\right]_{i,j  \in \{1,2\}}$.
\end{defn}
In the following, we will sometimes abuse notation and write $\psi_{\Lambda}$ instead of $\psi_{\Lambda,n}$, when the parameter $n$ is clear from context.

We note that the coherent states have a tensor product form in the sense that  
\begin{align*}
\psi_{\Lambda,n}=\psi_{\Lambda,1}^{\otimes n}.
\end{align*}
We will also write $|\Lambda,n\rangle = |\Lambda,1\rangle^{\otimes n}$ for $\psi_{\Lambda,n}$.
Such a state is called \emph{identically and independently distributed} (i.i.d.) in the quantum information literature.

The main feature of a family of coherent states is that they resolve the identity. This is the case with the $SU(2,2)$ coherent states introduced above: see Ref.~\cite{lev16}.

\begin{theorem}[Resolution of the identity]\label{thm:resol}
For $n \geq 4$, the coherent states resolve the identity over the symmetric subspace $F_{2,2,n}^{U(n)}$:
\begin{align*}
\int_{\D} |\Lambda,n\rangle \langle \Lambda,n| \mathrm{d}\mu_n(\Lambda) = \mathbbm{1}_{F_{2,2,n}^{U(n)}}, 
\end{align*}
where $\mathrm{d}\mu_n(\Lambda)$ is the invariant measure on $\D$ given by
\begin{align}\label{eqn:mu}
\mathrm{d}\mu_n (\Lambda) =  \frac{(n-1)(n-2)^2(n-3)}{\pi^{4}\det(\mathbbm{1}_2 - \Lambda \Lambda^\dagger)^4 } \prod_{i=1}^2 \prod_{j=1}^{2} \mathrm{d} \mathfrak{R}(\Lambda_{i,j}) \mathrm{d} \mathfrak{I}(\Lambda_{i,j}), 
\end{align}
where $\mathfrak{R}(\Lambda_{i,j})$ and $\mathfrak{I}(\Lambda_{i,j})$ refer respectively to the real and imaginary parts of $\Lambda_{i,j}$. This operator equality is to be understood for the weak operator topology.
\end{theorem}


\section{Technical lemmas}
\label{sec:lemmas}

In this section, we prove or recall a number of technical results that will be useful for analyzing the finite energy version of the de Finetti theorem in \ref{sec:finite}. 

\subsection{Tail bounds}

For positive integers $k,n >0$, the Beta and regularized (incomplete) Beta functions are given respectively by 
\begin{align*}
B(k,n) = \int_0^1 t^{k-1} (1-t)^{n-1} \d t = \frac{(k-1)!(n-1)!}{(n+k-1)!}, \quad B(x;k,n) = \int_0^x t^{k-1}(1-t)^{n-1} \d t,
\end{align*}
for $x >0$. 
Finally, the regularized Beta function is defined as
\begin{align*}
I_{x}(k,n) = \frac{B(x; k,n)}{B(k,n)}.
\end{align*}

Let us recall the Chernoff bound for a sum of independent Bernoulli variables. 
\begin{theorem}[Chernoff bound] \label{thm:chernoff}
Let $X_1, \ldots, X_n$ be independent random variables on $\{0,1\}$ with $\mathrm{Pr}[X_i=1]=p$, for $i=1, \ldots, n$. Set $X = \sum_{i=1}^n X_i$. Then for any $t \in [0,1-p]$, we have
\begin{align*}
\mathrm{Pr}[X \geq (p+t)n ] \leq \exp \left(-n D(p+t||p) \right),
\end{align*}
where the relative entropy is defined as $D(x||y) = x \log \frac{x}{y} + (1-x) \log \frac{1-x}{1-y}$.
\end{theorem}

Pinsker's inequality gives a lower bound on $D(x||y)$ as a function of the total variation distance between the two probability distributions.
\begin{lemma}[Pinsker's inequality]
\label{lem:pinsker}
For $0 < y < x <1$, it holds that
\begin{align*}
D(x \| y) \geq \frac{2}{\ln 2} (x-y)^2.
\end{align*}
\end{lemma}

We now prove a tail bound for the regularized Beta function.
\begin{lemma}[Tail bound for regularized Beta function]
\label{lem:tail-bound}
For integers $k,n >0$, it holds that 
\begin{align*}
1-I_{\eta}(k,n)  \leq \exp\left(-(n+k-1) D\left(\frac{k-2}{n+k-1} \| \eta \right) \right),
\end{align*}
provided that $\eta \geq (k-2)/(n+k-1)$.
\end{lemma}

\begin{proof}
The incomplete Beta function can be related to the tail of the binomial distribution as follows:
\begin{align}
1- I_{\eta}(k,n) &= F(k-1,n+k-1,\eta) \label{eqn:injected}
\end{align}
where $F(K,N,p)$ is the probability that there are at most $K$ successes when drawing $N$ times from a Bernoulli distribution with success probability $p$.
Equivalently, if $X_i$ are $\{0,1\}$-random variables such that $\mathrm{Pr}[X_i=1] = 1-p$ for $i = 1, \ldots, n$, then 
\begin{align*}
F(K,N,p) = \mathrm{Pr}[ X \geq N-K+1],
\end{align*}
where $X = \sum_{i=1}^{n} X_i$.
The Chernoff bound of Theorem \ref{thm:chernoff} yields
\begin{align*}
F(K,N,p) \leq \exp \left(-N D(1-p+t||1-p) \right)
\end{align*}
for $t = p - \frac{K-1}{N}$, provided that $N-K+1 \geq (1-p)N$, \textit{i.e.} $p \geq (K-1)/N$ or $\eta \geq (K-1)/N$.
Taking $K = k-1, N = n+k-1$ and $p=\eta$, and injecting into Eq.~\eqref{eqn:injected}, gives
\begin{align*}
1- I_{\eta}(k,n) & \leq \exp \left(-(n+k-1) D\left(1-\frac{k-2}{n+k-1}||1- \eta \right) \right)\\
& \leq \exp \left(-(n+k-1) D\left(\frac{k-2}{n+k-1}||\eta \right) \right),
\end{align*}
which holds provided that $\eta \geq (k-2)/(n+k-1)$. This proves the claim.
\end{proof}

\subsection{Energy cutoff}

The resolution of the identity of Theorem \ref{thm:resol} involves operators which are not trace-class, as well as coherent states with arbitrary large energy. The natural solution to get operators with finite norm is to replace the domain $\mathcal{D}$ by a cut-off versions $\mathcal{D}_\eta$ defined by
\begin{align*}
\mathcal{D}_\eta := \left\{ \Lambda \in M_{p,q} (\C) \: : \: \eta \1_p - \Lambda \Lambda^\dagger \geq 0\right\},
\end{align*}
for $\eta \in [0,1[$. Note that 
\begin{align*}
\lim_{\eta \to 1} \mathcal{D}_\eta = \mathcal{D}.
\end{align*}
The integration over $\mathcal{D}_\eta$ can then be performed by first integrating the measure $\d \mu_n(\Lambda)$ on the ``polar variables'', and only later on the ``radial'' variables corresponding to the singular values of $\Lambda$.

For a fixed pair of squared singular values $(x,y)$, let $V_{x,y}$ be the set of matrices $\Lambda \in \mathcal{D}$ with squared singular values $(x,y)$, \textit{i.e.},
\begin{align*}
V_{x,y} := \left\{u \big[\begin{smallmatrix} \sqrt{x} & 0 \\ 0 & \sqrt{y}\end{smallmatrix}\big] v^\dagger \: : \: u, v \in U(2)Ê\right\}.
\end{align*}
We further define the operator $P_{x,y}$ corresponding to the integral of $|\Lambda,n\rangle \langle \Lambda, n|$ over $V_{x,y}$: 
\begin{align}
P_{x,y} := \int_{V_{x,y}} |\Lambda,n\rangle\langle \Lambda,n| \d \mu_{x,y}(\Lambda) \geq 0 \label{eqn:Pxy}
\end{align}
where $\mathrm{d}\mu_{x,y} (\Lambda)$ is the Haar measure on $V_{x,y}$ and the normalization is chosen so that $\tr \, P_{x,y} = 1$.

We have the following equivalent version of the resolution of the identity of Theorem \ref{thm:resol}.
\begin{theorem}\label{thm:resol2}
For $n \geq 4$, it holds that:
\begin{align*}
\int_{0}^1 \int_0^1 q(x,y) P_{x,y} \d x \d y = \mathbbm{1}_{F_{2,2,n}^{U(n)}},
\end{align*}
where the distribution $q(x,y)$ is given by
\begin{align}
q(x,y) := \frac{(n-1)(n-2)^2(n-3) (x-y)^2}{2(1-x)^4 (1-y)^4}. \label{eqn:Qxy}
\end{align}
\end{theorem}

\begin{proof}
We wish to integrate $|\Lambda, n \rangle \langle \Lambda, n| \d\mu_n(\Lambda)$ over the ``polar'' variables. 
For this, we perform the singular value decomposition of $\Lambda$, which reads $\Lambda = u \Sigma v^\dagger$, where $u, v \in U(2)$ and $\Sigma = \mathrm{diag}(\sigma_1, \sigma_2)$, with $\sigma_1, \sigma_2 \in [0,1[$.

The Jacobian for the singular value decomposition is \cite{mui82}:
\begin{align}
\mathrm{d}\Lambda = (\sigma_1^2-\sigma_1^2)^2 \sigma_1 \sigma_2 (u^\dagger \mathrm{d} u) \mathrm{d} \Sigma (v^\dagger \mathrm{d} v).
\end{align}

Exploiting this Jacobian and performing the change of variables $x = \sigma_1^2$, $y=\sigma_2^2$, one obtains that the resolution of the identity of Theorem \ref{thm:resol} can be written: 
\begin{align*}
C \int_0^1 \d x \int_0^1 \d y \frac{(x-y)^2}{(1-x)^4(1-y)^4} P_{x,y} = \mathbbm{1}_{F_{2,2,n}^{U(n)}},
\end{align*}
for the appropriate constant $C$. Here, we have used that $\det(\mathbbm{1}_2 - \Lambda \Lambda^\dagger)^4 = (1-x)^4 (1-y)^4$ for any $\Lambda \in V_{x,y}$. 

The constant $C$ can be determined by considering the overlap between $\1_{F_{2,2,n}^{U(n)}}$ and the vacuum state:
\begin{align*}
1 & = \langle 0 | \1_{F_{2,2,n}^{U(n)}}| 0\rangle\\
 &= C \int_0^1 \d x \int_0^1 \d y \int \frac{(x-y)^2}{(1-x)^4(1-y)^4} \big \langle 0 \big|u \big[\begin{smallmatrix} \sqrt{x} & 0 \\ 0 & \sqrt{y}\end{smallmatrix}\big] v^\dagger, n\big\rangle \big\langle u  \big[\begin{smallmatrix} \sqrt{x} & 0 \\ 0 & \sqrt{y}\end{smallmatrix}\big] v^\dagger, n\big|0 \big\rangle \d u \d v\\
 &= C \int_0^1 \d x \int_0^1 \d y\frac{(x-y)^2}{(1-x)^4(1-y)^4}(1-x)^n (1-y)^n \int  \d u \d v\\
 &= C \int_0^1 \d x \int_0^1 \d y\frac{(x-y)^2}{(1-x)^4(1-y)^4}(1-x)^n (1-y)^n\\
 &= C \frac{2}{(n-1)(n-2)^2(n-3)},
 \end{align*} 
 where we used that $ \big \langle 0 \big|u \big[\begin{smallmatrix} \sqrt{x} & 0 \\ 0 & \sqrt{y}\end{smallmatrix}\big] v^\dagger, n\big\rangle = (1-x)^{n/2} (1-y)^{n/2}$ for any $u, v\in U(2)$ and that the measures $\d u$ and $\d v$ are normalized.
\end{proof}

Let $K \geq 0$ be an integer. We define $V_{=K}$ as the subspace of $F_{2,2,n}^{U(n)}$ spanned by vectors with $K$ pairs of excitations:
\begin{align*}
V_K := \mathrm{Span}\{Z_{1,1}^i Z_{1,2}^j Z_{2,1}^k Z_{2,2}^\ell  |0\rangle \: : \: i + j + k+\ell = K; i, j,k, \ell \in \N \},
\end{align*}
and the projector $\Pi_{=K}$ to be the orthogonal projector onto $V_{=K}$. Physically, this is the subspace of the Fock space restricted to states containing $2K$ photons in total in the $4n$ optical modes. 

Moreover, let us denote by $a_k^n := \tbinom{n+k-1}{k}$ the number of configurations of $k$ particles in $n$ modes.

\begin{lemma}
\label{lem:Piq} For $K \in \N$ and $x, y \in [0,1[$, we have
\begin{align*}
\tr \left[\Pi_{=K} P_{x,y} \right] = \sum_{k_1+k_2=K} a_{k_1}^n a_{k_2}^n (1-x)^n (1-y)^n x^{k_1} y^{k_2}.
\end{align*}
\end{lemma}

\begin{proof}
The total photon number distribution of a state $|\Lambda, n\rangle$ is invariant under local unitaries $u, v \in U(2)$ applied on the creation operators of $F_A$ or $F_B$. This means that this distribution only depends on the squared singular values of the matrix $\Lambda$. 
In particular, denoting by $|(\sqrt{x},\sqrt{y}),n\rangle$ the coherent state corresponding to the matrix $\mathrm{diag}(\sqrt{x}, \sqrt{y})$, we obtain:
\begin{align*}
\tr \left[\Pi_{=K} P_{x,y} \right] = \langle (\sqrt{x},\sqrt{y}),n |  \Pi_{=K} | (\sqrt{x},\sqrt{y}),n\rangle.
\end{align*}
Since this coherent state is given by
\begin{align*}
| (\sqrt{x},\sqrt{y}),n\rangle := (1-x)^{n/2} (1-y)^{n/2} \exp(\sqrt{x} Z_{11}) \exp(\sqrt{y} Z_{22}),
\end{align*}
it implies that
\begin{align*}
\tr \left[\Pi_{=K} P_{x,y} \right] =\sum_{k_1+k_2=K} a_{k_1}^n a_{k_2}^n (1-x)^n (1-y)^n x^{k_1} y^{k_2}.
\end{align*}
\end{proof}

Let us define the operator $\overline{P}_{\eta}$ as
\begin{align*}
\overline{P}_{\eta}  := \int_{\mathcal{D} \setminus \mathcal{D}_\eta} |\Lambda,n\rangle \langle \Lambda,n| \d \mu_n(\Lambda).
\end{align*}

\begin{lemma}
\label{eqn:crucial-step} 
For $n \geq 38$, $K \in \N$ and $\eta \in [0,1[$ such that $K \leq \frac{\eta}{1-\eta}(n-5) $, it holds that
\begin{align*}
\tr (\Pi_{=K} \overline{P}_\eta) \leq   2 N^4 (1+\alpha)^7 \exp\left(-N D\left(\frac{\alpha}{\alpha+1} \big\|\eta \right) \right),
\end{align*}
where $N = n-5$ and $\alpha := K/N$.
\end{lemma}

\begin{proof}
For any non negative distribution $f(x,y) \geq 0$ symmetric in $x$ and $y$, \textit{i.e.} such that $f(x,y) = f(y,x)$, it holds that
\begin{align*}
\int_{\overline{\mathcal{E}}_\eta} f(x,y) \d x  \d y &\leq  2\int_{\eta}^1 \d x \int_0^1 \d y f(x,y).
\end{align*}
Since $q(x,y) \tr \left[P_{=K} \Pi_{x,y} \right]$ is such a distribution, it holds that
\begin{align*}
\tr (\Pi_{=K} \overline{P}_\eta)\leq 2\int_{\eta}^1 \d x \int_0^1 \d y  q(x,y) \, \tr \left[P_{=K} \Pi_{x,y} \right].
\end{align*}
Lemma \ref{lem:Piq} then yields
\begin{align*}
\tr (\Pi_{=K} \overline{P}_\eta) &\leq 2  \sum_{k_1+k_2=K} a_{k_1}^n a_{k_2}^n \int_{\eta}^1 \d x (1-x)^n x^{k_1} \int_0^1 \d y   q(x,y)  (1-y)^n  y^{k_2}\\
&\leq  (n-1)(n-2)^2(n-3)  \sum_{k_1+k_2=K} a_{k_1}^n a_{k_2}^n \int_{\eta}^1 \d x (1-x)^{n-4} x^{k_1} \int_0^1 \d y  (x-y)^2  (1-y)^{n-4}  y^{k_2}\\
&\leq  (n-1)(n-2)^2(n-3)  \sum_{k_1+k_2=K} a_{k_1}^n a_{k_2}^n \int_{\eta}^1 \d x (1-x)^{n-4} x^{k_1} \int_0^1 \d y   (1-y)^{n-4}  y^{k_2}
\end{align*}
where we used the trivial bound $ (x-y)^2 \leq 1$ for $0 \leq x,y \leq 1$ in the last inequality.

The normalization of the Beta function reads
\begin{align*}
\int_0^1  (1-y)^{n} y^k \d y =\frac{k! n!}{(n+k+1)!},
\end{align*}
which gives 
\begin{align}\label{eqn:interm}
\tr (\Pi_{=K} \overline{P}_\eta) &\leq  (n-2)  \sum_{k_1+k_2=K} (n+k_2-1)(n+k_2-2)   \int_{\eta}^1 a_{k_1}^n (1-x)^{n-4} x^{k_1}  \d x
\end{align}
Lemma \ref{lem:tail-bound} allows us to bound the integral: 
\begin{align}
\int_{\eta}^1 a_{k_1}^{n-4} (1-x)^{n-4} x^{k_1}  \d x \leq\exp\left(-(n+k_1-5) D\left(\frac{n-3}{n+k_1-5} \big\| 1-\eta \right) \right),
\end{align}
provided that $1-\eta \leq (n-3)/(n+k_1-5)$.

If $\frac{n-5}{n+K-5} \geq 1-\eta$, this term can be bounded uniformly as 
\begin{align*}
\int_{\eta}^1 a_{k_1}^{n-4} (1-x)^{n-4} x^{k_1}  \d x \leq\exp\left(-N D\left(\frac{N}{N+K} \big\| 1-\eta \right) \right),
\end{align*}
where we defined $N := n-5$. Injecting this in Eq.~\eqref{eqn:interm}, we obtain
\begin{align}
\tr (\Pi_{=K} \overline{P}_\eta) &\leq  (N+3)  \sum_{k_1+k_2=K} (N+k_2+4)(N+k_2+3) \frac{(N+k_1+4)! N!}{(N+k_1)! (N+4)!}  \exp\left(-N D\left(\frac{N}{N+K} \big\| 1-\eta\right) \right) \nonumber\\
&\leq   \frac{(K+1) (N+K+4)^6}{(N+1)^3}  \exp\left(-N D\left(\frac{N}{N+K} \big\| 1-\eta \right) \right) \label{eqn:bound33}
\end{align}
Imposing in addition that $N\geq 4$, \textit{i.e.} $n\geq 9$, so that $N+K+4 \leq 2(N+K)$, one finally obtains the bound:
\begin{align*}
\tr (\Pi_{=K} \overline{P}_\eta) \leq   64 \frac{(N+K)^7}{N^3}  \exp\left(-N D\left(\frac{N}{N+K} \big\| 1-\eta \right) \right).
\end{align*}
One can get a better bound by choosing $N \geq 33$, \textit{i.e.} $n\geq 38$: in that case, one can check that for any $K \geq 0$, it holds that
\begin{align*}
\left( 1 + \frac{4}{N+K}\right)^6 \leq 2,
\end{align*}
which gives $(N+K+4)^6 \leq 2(N+K)^6$. Injecting this into Eq.~\eqref{eqn:bound33} yields
\begin{align*}
\tr (\Pi_{=K} \overline{P}_\eta) \leq   2 \frac{(N+K)^7}{N^3}  \exp\left(-N D\left(\frac{N}{N+K} \big\| 1-\eta \right) \right).
\end{align*}
\end{proof}

\begin{lemma}
\label{lem:proj}
For any nonnnegative operator $A \geq 0$ and projector $\Pi$ with $\mathrm{rank}(\Pi) < \infty$, it holds that: 
$$\Pi A \Pi \leq \tr[\Pi A ] \Pi.$$ 
\end{lemma}

\begin{proof}
The support of $\Pi A \Pi$ is contained in that of $\tr[\Pi A ] \Pi$. 
Since both operators are positive semi-definite, the only thing we need to prove is that for any $\lambda \in \mathrm{spec}(\Pi A \Pi)$, it holds that
\begin{align*}
\lambda \leq \tr[\Pi A]
\end{align*}
since all the nonzero eigenvalues of $\tr[\Pi A ] \Pi$ are equal to $\tr[\Pi A ]$.
The sum of the eingenvalues of an operator is equal to its trace, which gives
\begin{align*}
\sum_{\lambda \in \mathrm{spec}(\Pi A \Pi)}Ê\lambda = \tr (\Pi A \Pi).
\end{align*}
Moreover, since all these eigenvalues are nonnegative, we have that $\lambda_{\max} (\Pi A \Pi) \leq \sum_{\lambda \in \mathrm{spec}(\Pi A \Pi)}Ê\lambda$, which concludes the proof.
\end{proof}


\section{Finite energy version of de Finetti theorem}
\label{sec:finite}

In this section, we establish a \emph{de Finetti reduction}, similar to the one obtained in \cite{CKR09} in the case of permutation invariance. 
Such a reduction uses as a main tool as statement analogous to the resolution of the identity
\begin{align*}
\mathbbm{1}_{\mathrm{Sym}} \leq C_{n,d} \int \left(|\phi\rangle \langle \phi|\right)^{\otimes n} \d \mu(\phi)
\end{align*}
where $C(n,d)$ is a polynomial in $n$ provided the local dimension $d$ is finite. 

In the case of continuous-variable protocols, the local dimension is infinite and we need to find a better reduction. This is indeed possible provided we have bounds on the maximum energy (or total number of photons) of the states under consideration. 

For $\eta \in [0,1[$, define the sets $\mathcal{E}_\eta = [0, \eta] \times [0,\eta]$ and $\overline{\mathcal{E}}_\eta = [0,1[^2 \setminus \mathcal{E}_\eta$. 

We introduce the following positive operators 
\begin{align}
P_\eta &:=  \int_{\mathcal{E}_\eta} q(x,y) P_{x,y} \d x \d y =  \int_{\mathcal{D}_\eta} |\Lambda,n\rangle \langle \Lambda,n| \d \mu_n(\Lambda), \label{eqn:Peta} \\
\overline{P}_{\eta} &:= \int_{\overline{\mathcal{E}}_\eta} q(x,y) P_{x,y} \d x \d y =  \int_{\mathcal{D} \setminus \mathcal{D}_\eta} |\Lambda,n\rangle \langle \Lambda,n| \d \mu_n(\Lambda),
\end{align}
where the equalities follow from the fact that one can integrate over $\mathcal{D}_\eta$ by first integrating over $V_{x,y}$ and then over $\mathcal{E}_\eta$. We recall that the operator $P_{xy}$ is defined in Eq.~\eqref{eqn:Pxy} and that the distribution $q(x,y)$ is defined in Eq.~\eqref{eqn:Qxy}.
The resolution of the identity over $F_{2,2,n}^{U(n)}$ (Theorem \ref{thm:resol}) immediately implies that
\begin{align*}
P_\eta + \overline{P}_\eta = \1_{F_{2,2,n}^{U(n)}}.
\end{align*}

Let $K \geq 0$ be an integer. We recall that $V_{=K}$ is the subspace of $F_{2,2,n}^{U(n)}$ spanned by vectors with $K$ pairs of excitations:
\begin{align*}
V_K := \mathrm{Span}\{Z_{1,1}^i Z_{1,2}^j Z_{2,1}^k Z_{2,2}^\ell  |0\rangle \: : \: i + j + k+\ell = K; i, j,k, \ell \in \N \}.
\end{align*}
The subspace $V_{\leq K}$ is defined as $V_{\leq K} := \bigoplus_{k=0}^K V_{=k}$.
The projector $\Pi_{=K}$ is the orthogonal projector onto $V_{=K}$ and the projector $\Pi_{\leq K}$ is defined as
\begin{align*}
\Pi_{\leq K} : = \sum_{k=0}^K \Pi_{=k}.
\end{align*}

\begin{theorem}[Finite energy version of de Finetti theorem (Theorem 1 from the main text)] \label{thm:finite-version} 
For $n\geq 5$ and $\eta \in [0,1[$, if $K \leq\frac{\eta}{1-\eta} (n-5) $, then the following operator inequality holds 
\begin{align*}
\int_{\mathcal{D}_\eta} |\Lambda,n\rangle \langle \Lambda,n | \d \mu_n(\Lambda)\geq (1-\eps) \Pi_{\leq K}
\end{align*}
with
\begin{align*}
\eps := 2 N^4 (1+\alpha)^7 \exp\left(-N D\left(\frac{\alpha}{\alpha+1} \big\| \eta \right) \right).
\end{align*}
for $\alpha = K/N$ and $N=n-5$.
\end{theorem}

In particular, choosing $K$ such that $\alpha = \frac{1+\eta}{1-\eta} = \frac{K}{N}$ and using Pinsker's inequality (Lemma \ref{lem:pinsker}) yields
\begin{align*}
\eps \leq  \frac{2 (N+K)^7}{N^3}\exp\left(- \frac{2N^3}{(N+K)^2 \ln 2} \right).
\end{align*}

\begin{proof}[Proof of Theorem \ref{thm:finite-version}]

The resolution of the identity reads
\begin{align*}
\int_{\overline{\mathcal{E}}_\eta}  P_{x,y}q(x,y) \d x \d y+\int_{{\mathcal{E}}_\eta} P_{x,y}q(x,y) \d x \d y = \1_{F_{2,2,n}^{U(n)}} = \sum_{k=0}^\infty\Pi_{=k} .
\end{align*}
For all $k \leq K$, the projector $\Pi_{=k}$ can be written as:
\begin{align*}
\Pi_{=k} = \int_{\overline{\mathcal{E}}_\eta} \Pi_{=k}P_{x,y}\Pi_{=k} q(x,y) \d x \d y+\int_{{\mathcal{E}}_\eta} \Pi_{=k}  P_{x,y} \Pi_{=k} q(x,y) \d x \d y.
\end{align*}
In particular, since $k \leq K \leq \frac{\eta}{1-\eta} (n-5)$, we have
\begin{align}
\int_{{\mathcal{E}}_\eta}  P_{x,y}q(x,y) \d x \d y &\geq \int_{{\mathcal{E}}_\eta}   \Pi_{=k}P_{x,y} \Pi_{=k} q(x,y) \d x \d y \nonumber\\
&\geq \Pi_{=k}   - \int_{\overline{\mathcal{E}}_\eta}  \Pi_{=k}P_{x,y} \Pi_{=k} q(x,y) \d x \d y \nonumber \\
&\geq  \Pi_{=k}  - \int_{\overline{\mathcal{E}}_\eta}  \tr[ \Pi_{=k}P_{x,y}]  \Pi_{=k} q(x,y) \d x \d y  \label{eqn640}\\
&\geq  (1-\eps)\Pi_{=k}     \label{eqn641}
\end{align}
where we used Lemma \ref{lem:proj} in Eq.~\eqref{eqn640} and the upper bound resulting from Lemma \ref{eqn:crucial-step}:
\begin{align*}
\int_{\overline{\mathcal{E}}_\eta} \tr \left[ \Pi_{=k}P_{x,y}\right] q(x,y) \d x \d y \leq \eps
\end{align*}
in Eq.~\eqref{eqn641}.
It follows that:
\begin{align*}
\int_{{\mathcal{E}}_\eta}  P_{x,y}q(x,y) \d x \d y \geq (1-\eps)  \Pi_{=k} 
\end{align*}
for all $k \leq K$.
This finally implies that 
\begin{align*}
\int_{{\mathcal{E}}_\eta} P_{x,y}q(x,y) \d x \d y \geq(1-\eps)  \sum_{k \leq K} \Pi_{=k}   = (1-\eps)  \Pi_{\leq K}.
\end{align*}
\end{proof}

The crucial property of Theorem \ref{thm:finite-version} that will be important for application is that the volume of $\mathcal{D}_\eta$ is finite, and scales as a low degree polynomial in $n$ and $K$.

\begin{theorem}\label{thm:volume}
For $n \geq 38$, $K \geq n-5$ and $\eta = \frac{K-n+5}{K+n-5}$, it holds that 
\begin{align*}
T(n,\eta) := \tr \int_{\mathcal{D}_\eta} |\Lambda,n\rangle\langle \Lambda,n| \d \mu_n(\Lambda) \leq \frac{K^4}{100}.
\end{align*}
\end{theorem}

\begin{proof}
The volume of $\mathcal{D}_\eta$ is given by
\begin{align*}
\tr \int_{\mathcal{D}_\eta} |\Lambda,n\rangle\langle \Lambda,n| \d \mu_n(\Lambda) &= \int_0^{\eta}  \int_0^{\eta}  q(x,y)\d x \d y \\
&= \frac{(n-1)(n-2)^2(n-3)\eta^4}{12(1-\eta)^4} \\
&\leq \frac{n^4 \eta^4}{12 (1-\eta)^4}\\
& \leq \frac{n^4(1+\eta)^4}{192(1-\eta)^4}
\end{align*}
where we used that $\eta \leq (1+\eta)/2$ in the last equation.

In particular, choosing $K$ such that $\eta = \frac{K-n+5}{K+n-5}$ gives $\frac{1+\eta}{1-\eta} = \frac{K}{n-5}$, and therefore
\begin{align*}
\tr \int_{\mathcal{D}_\eta} |\Lambda,n\rangle\langle \Lambda,n| \d \mu_n(\Lambda) \leq \frac{n^4 K^4}{192(n-5)^4}.
\end{align*}

For $n \geq 38$, it holds that $\frac{1}{192} \left(\frac{n}{n-5}\right)^4 \leq \frac{1}{100}$, which finally gives
\begin{align*}
\tr \int_{\mathcal{D}_\eta} |\Lambda,n\rangle\langle \Lambda,n| \d \mu_n(\Lambda) \leq \frac{K^4}{100}.
\end{align*}
\end{proof}
For future reference, let us not that that under the same assumptions as in the theorem, the following inequality also holds:
\begin{align}
T(n,\eta) +1 \leq \frac{K^4}{100}. \label{eqn:tighter}
\end{align}
This inequality will later be useful to analyze Eq.~\eqref{eqn:final-result} at the end of Section \ref{sec:general-proof} and obtain Eq.~(5) in the main text.

In other words, the volume $T(n,\eta)$ of $\cD_\eta$ is upper bounded by a polynomial of degree 4 in the number of modes (or equivalently in the total energy).

Let us define the function $N^* : [1,\infty[ \to \N$ such that 
\begin{align}
N^*(\alpha) = \max\left\{ 38,  \min \left\{N \in \N \: : \: 2(1+\alpha)^7 N^4 \exp\left(- \frac{2N}{(1+\alpha)^2 \ln 2} \right) \leq \frac{1}{2} \right\}\right\}. \label{eqn:N*}
\end{align}
For instance, $N^*(21) \approx 10^4$, $N^*(60) \approx 10^5$. 


We obtain:
\begin{corol}\label{corol16}
For $K \geq n-5$, if $n \geq N^*\big( \frac{K}{n-5} \big)-5$ then, for $\eta^* = \frac{K-n+5}{K+n-5}$, it holds that 
\begin{align*}
&\int_{\mathcal{D}_{\eta^*}} |\Lambda,n\rangle \langle \Lambda,n | \d \mu_n(\Lambda)\geq \frac{1}{2} \Pi_{\leq K} \\
&\tr \int_{\mathcal{D}_{\eta^*}} |\Lambda,n\rangle\langle \Lambda,n| \d \mu_n(\Lambda) \leq \frac{K^4}{100}.
\end{align*}
\end{corol}


\section{Security proof for a modified CV QKD protocol}
\label{sec:general-proof}

In this section, we recall some facts about security proofs for QKD protocols and explain how to obtain a secure protocol from an initial protocol $\cE_0$ known to be secure against Gaussian collective attacks, by prepending an energy test and adding an additional privacy amplification test. These various steps will then be detailed in the subsequent sections.

{\bf QKD protocols and their security}.--- A QKD protocol is a CP map from the infinite-dimensional Hilbert space $(\mathcal{H}_A\otimes \mathcal{H}_B)^{\otimes n}$, corresponding to the initially distributed entanglement,  to the set of pairs $(S_A,S_B)$ of $\ell$-bit strings (Alice and Bob's final keys, respectively) and $C$, a transcript of the classical communication.
In order to assess the security of a given QKD protocol $\mathcal{E}$ in a composable framework, one compares it with an ideal protocol \cite{MKR09,PR14}. 
The action of an ideal protocol $\mathcal{F}$ is defined by concatenating the protocol $\mathcal{E}$ with a map $\mathcal{S}$ taking $(S_A,S_B,C)$ as input and outputting the triplet $(S,S,C)$ where the string $S$ is a perfect secret key (uniformly distributed and unknown to Eve) with the same length as $S_A$, that is $\mathcal{F} = \mathcal{S}\circ\mathcal{E}$. Then, a protocol will be called \emph{$\epsilon$-secure} if the advantage in distinguishing it from an ideal version is not larger than $\epsilon$. This advantage is quantified by (one half of) the  diamond norm defined by
\begin{equation}
||\mathcal{E} - \mathcal{F}||_\diamond := \sup_{\rho_{ABE} } \left\|(\mathcal{E}-\mathcal{F})\otimes \mathrm{id}_\mathcal{K} (\rho_{ABE})\right\|_1,
\end{equation}
where the supremum is taken over density operators on $(\mathcal{H}_A\otimes \mathcal{H}_B)^{\otimes n} \otimes \mathcal{K}$ for any auxiliary system $\mathcal{K}$. 
The diamond norm is also known as the \emph{completely bounded trace norm} and quantifies a notion of distinguishability for quantum maps \cite{wat16}.

Our main technical result is a reduction of the security against general attacks to that against Gaussian collective attacks, for which security has already been proved in earlier work, for instance in \cite{Lev15}. 
Let us therefore suppose that our CV QKD protocol of interest, $\mathcal{E}_0$, is secure against Gaussian collective attacks. We will slightly modify it by prepending an initial test $\mathcal{T}$. More precisely, $\mathcal{T}$ is a CP map taking a state in a slightly larger Hilbert space, $(\mathcal{H}_A\otimes \mathcal{H}_B)^{\otimes (n+k)}$, applying a random unitary $u \in U(n+k)$ to it (corresponding to a network of beamsplitters and phaseshifters), measuring the last $k$ modes and comparing the measurement outcome to a threshold fixed in advance. 
The test succeeds if the measurement outcome (related to the energy) is small, meaning that the global state is compatible with a state containing only a low number of photons per mode. Such a state is well-described in a low dimensional Hilbert space, as we will discuss in Section \ref{sec:test}.
Depending on the outcome of the test, either the protocol aborts, or one applies the original protocol $\mathcal{E}_0$ on the $n$ remaining modes. 

For the test to be practical, it is important that the legitimate parties do not have to physically implement the transformation $u \in U(n+k)$. Rather, they can both measure their $n+k$ modes with heterodyne detection, perform a random rotation of their respective classical vector in $\R^{2(n+k)}$ according to $u \in U(n+k) \cong O(2(n+k)) \cap Sp(2(n+k))$.

In this paper, we assume that this symmetrization step is performed, as it is anyway required for the security proof of the protocol against collective attacks \cite{Lev15}. 
We believe, however, that this step might not be required for establishing the security of the protocol and leave it as an important open question for future work. In particular, recent proof techniques in discrete-variable QKD have shown that the permutation need not be applied in practice \cite{TLG12}.

 In section \ref{sec:generalization}, we will prove a de Finetti reduction that allows to upper bound the diamond distance between two quantum channels, provided that they display the right invariance under the action of the unitary group $U(n)$ and that the input states have a maximum number of photons. 
We address this second issue by introducing another CP map $\mathcal{P}$ which projects a state acting on $F_{1,1,n} = (\mathcal{H}_A\otimes \mathcal{H}_B)^{\otimes n}$ onto a low-dimensional Hilbert space $F_{1,1,n}^{\leq K} $ with less than $K$ photons overall in the $2n$ modes shared by Alice and Bob. 
Here, the value of $K$ scales linearly with $n$.

Let us denote by $\cE_0$ a CV QKD proven ${\eps}$-secure against Gaussian collective attacks, for instance as in \cite{Lev15}. 
This means that (see Section \ref{sec:collective} for details)
\begin{align}
\|((\cE_0 - \cF_0)\otimes \1) (|\Lambda,n\rangle \langle \Lambda, n|)\|_1 \leq {\eps},
\end{align}
for any generalized coherent state $|\Lambda,n\rangle$. Here $\cF_0 := \cS \circ \cE_0$ and $\cS$ is a map that replaces the output key of $\cE_0$ by an independent and uniformly distributed string of length $\ell$ when $\cE_0$ did not abort, and does nothing otherwise.

Here $\cE_0$ maps an arbitrary density operator $\rho_{AB} \in \mathfrak{S}(F_{1,1,n})$ to a state $\rho_{S_A, S_B, C}$ where the registers are all classical and store respectively Alice's final key, Bob's final key and a transcript of the classical communication.

Let us define the following maps:
\begin{align*}
\cT&: \B(F_{1,1,n+k}) \to \B(F_{1,1,n}) \otimes \{\mathrm{passes} / \mathrm{aborts}\},\\ 
\cP&: \B(F_{1,1,n}) \to \B(F_{1,1,n}^{\leq K}),\\ 
\cR &: \{0,1\}^{\ell} \times \{0,1\}^{\ell} \to \{0,1\}^{\ell'} \times \{0,1\}^{\ell'},
\end{align*}
where 
\begin{itemize}
\item $\cT(k, d_A, d_B)$ takes as input an arbitrary state $\rho_{AB}$ on $F_{1,1,n+k}$, maps it to $V_u \rho_{AB} V_u^{\dagger}$ where the unitary $u$ is chosen from the Haar measure on $U(n+k)$, measures the last $k$ modes for $A$ and $B$ with heterodyne detection and check whether the measurement outputs pass the test if the $k$ outcomes $\alpha_1, \cdots, \alpha_k$ of Alice and $\beta_1, \cdots, \beta_k$ of Bob satisfy
\begin{align*}
\sum_{i=1}^k |\alpha_i|^2 \leq k d_A \quad \text{and} \quad \sum_{i=1}^k |\beta_i|^2 \leq k d_B.
\end{align*} 
If they pass the test, the map returns the state on the first $n$ modes (that were not measured) as well as the flag ``passes''. Otherwise, it returns the vacuum state and the flag ``aborts''.
\item $\cP$ is the projector onto the finite-dimensional subspace $F_{1,1,n}^{\leq K}$ (corresponding to states with at most $K$ photons in the $2n$ modes): it maps any state $\rho \in \B(F_{1,1,n})$ to $\Pi_{\leq K} \rho \Pi_{\leq K} \in \B(F_{1,1,n}^{\leq K})$. This trace non-increasing map is introduced as a technical tool for the security analysis but need not be implemented in practice. It simply ensures that the states that are fed to the original QKD protocol $\cE_0$ live in a finite-dimensional subspace.
In the text, we will alternatively denote this projection by $\cP^{\leq K}$ or $\cP(n,K)$, depending on which parameters we wish to make explicit.
\item $\cR$ takes two $\ell$-bit strings as input and returns $\ell'$-bit strings (for $\ell' < \ell$).
\end{itemize}

We finally define our CV QKD protocol $\cE$ as
\begin{align*}
\cE =  \cR \circ \cE_0 \circ \cT
\end{align*}
and the ideal protocol as $\cF = \cS \circ \cE$. Abusing notation slightly, the map $\cS$ now acts on strings of length $\ell'$ instead of $\ell$.

\begin{lemma}\label{lem:sec-red}
Let $\overline{\cE}$ be the protocol $\cR \circ \cE_0$ where the inputs are restricted to the finite-dimensional subspace $\B(F_{1,1,n}^{\leq K})$, and $\overline{\cF} = \cS \circ \overline{\cE}$. Then the security of $\overline{\cE}$ implies the security of $\cE$:
\begin{align} \label{eqn:sec-red}
 ||\mathcal{E} - \mathcal{F}||_\diamond &\leq ||\overline{\mathcal{E}} - \overline{\mathcal{F}}||_\diamond  + 2 || (\1- \mathcal{P}) \circ \mathcal{T}||_\diamond,
\end{align}
provided that the quantity $|| (\1- \mathcal{P}) \circ \mathcal{T}||_\diamond$ can be made arbitrarily small. 
\end{lemma}

\begin{proof}
We define (virtual) protocols $\tilde{\mathcal{E}}:=  \cR \circ  \mathcal{E}_0 \circ \mathcal{P} $ and $\tilde{\mathcal{F}}:=  \mathcal{S} \circ \tilde{\mathcal{E}}$. 
The security of the protocol $\mathcal{E}$ is then a consequence of the following derivation:
\begin{align}
 ||\mathcal{E} - \mathcal{F}||_\diamond &\leq   ||\tilde{\mathcal{E}}\circ \cT - \tilde{\mathcal{F}} \circ \cT||_\diamond  +  ||\mathcal{E} - \tilde{\mathcal{E}} \circ \cT||_\diamond+ ||\mathcal{F} - \tilde{\mathcal{F}} \circ \cT||_\diamond \nonumber \\
&\leq   ||(\tilde{\mathcal{E}} - \tilde{\mathcal{F}}) \circ \cT||_\diamond  +  ||\cR \circ \mathcal{E}_0  \circ (\mathrm{id}- \mathcal{P}) \circ \mathcal{T}||_\diamond +||\cS \circ \cR \circ \mathcal{E}_0  \circ (\mathrm{id}- \mathcal{P}) \circ \mathcal{T}||_\diamond \nonumber\\
&\leq ||\tilde{\mathcal{E}} - \tilde{\mathcal{F}}||_\diamond  + 2 || (\1- \mathcal{P}) \circ \mathcal{T}||_\diamond ,
\end{align} 
where we used the triangle inequality and the fact that the CP maps $\cT$, $\cR\circ \mathcal{E}_0$ and $\mathcal{S}$ cannot increase the diamond norm. 

Since $\overline{E} \circ \cP = \tilde{E}$ and $\cP$ is trace non-increasing, we finally obtain that
\begin{align*}
 ||\mathcal{E} - \mathcal{F}||_\diamond &\leq ||\overline{\mathcal{E}} - \overline{\mathcal{F}}||_\diamond  + 2 || (\1- \mathcal{P}) \circ \mathcal{T}||_\diamond.
\end{align*}
\end{proof}

Bounding the two terms in the right hand side of Eq.~\eqref{eqn:sec-red} is done with the two following theorems, which will be proven in Sections \ref{sec:collective} and \ref{sec:test}, respectively. 

\begin{theorem}\label{thm:diamond-protocol}
With the previous notations, if $\cE_0$ is $\eps$-secure against Gaussian collective attacks, then 	
\begin{align*}
 ||\overline{\mathcal{E}} - \overline{\mathcal{F}}||_\diamond  \leq 2 T(n,\eta)  \eps
\end{align*}
where $T(n,\eta) =(n-1)(n-2)^2(n-3) \frac{\eta^4}{12(1-\eta)^4}$ and $\overline{\cE} = \cR \circ \cE_0 \circ \cP^{\leq K}$.
\end{theorem}

\begin{theorem}\label{thm:test}
For integers $n,k \geq 1$, and $d_A, d_B >0$, define $K = n(d'_A + d'_B)$ for $d'_{A/B} = d_{A/B} g(n,k,\eps/4)$ for the function $g$ defined in Eq.~\eqref{eqn:g}. Then
\begin{align*}
\big\| \big(\1 - \cP(n,K)\big) \circ \cT(k, d_A, d_B)\big\|_{\diamond} \leq \eps.
\end{align*}
\end{theorem}

Putting everything together yields our main result. 
\begin{theorem}\label{thm:main}
If the protocol $\cE_0$ is $\eps$-secure against Gaussian collective attacks, then the protocol $\cE = \cR \circ \cE_0 \circ \cP$ is $\eps'$-secure against general attacks with 
\begin{align*}
\eps' \leq 2 T(n,\eta) \eps + \eps.
\end{align*}
\end{theorem}

Putting everything together, we show that if $\cE_0$ is covariant under the action of the unitary group and $\eps$-secure against Gaussian collective attacks, then the protocol $\cE = \cR \circ \cE_0 \circ \cT$ is $\eps'$-secure against general attacks, with
\begin{align}
\eps' = 2\eps( T(n, \eta)+1) \label{eqn:final-result}
\end{align}
for $T(n,\eta) \leq \frac{1}{12} \left(\frac{\eta n}{1-\eta}\right)^4$, $\eta = \frac{K-n+5}{K+n-5}$ and $K = n(d_A + d_B)\left(1 + 2 \sqrt{\frac{\ln (8/\eps)}{2n}} +  \frac{\ln (8/\eps)}{n}\right)\left(1-2{\sqrt{\frac{\ln (8/\eps)}{2k}}}\right)^{-1}$. 
The first term in Eq.~\eqref{eqn:final-result} results from the de Finetti reduction and the second term results for the energy test failure probability.

In particular, for $n \geq 38$ and $K \geq n-5$, we obtain the bound $T(n,\eta) +1 \leq \frac{K^4}{100}$ from Eq.~\eqref{eqn:tighter}.
This yields $\eps' \leq \frac{K^4}{50} \eps$, which corresponds to Eq.~(5) in the main text.


\section{Generalization of the postselection technique of Ref.~\cite{CKR09}}
\label{sec:generalization}

The goal of this section is to prove the following theorem (Theorem 2 in the main text).

\begin{theorem}[Theorem 2 of the main text] \label{thm:postselection}
Let $\Delta: \mathrm{End}(F_{1,1,n}^{\leq K}) \to  \mathrm{End}(\mathcal{H}')$ such that for all $u \in U(n)$, there exists a CPTP map $\mathcal{K}_u: \mathrm{End}(\mathcal{H}') \to \mathrm{End}(\mathcal{H}')$ such that $\Delta \circ u = \mathcal{K}_u: \circ \Delta$, then 
\begin{align*}
\|\Delta\|_\diamond \leq 2 T(n,\eta) \|(\Delta \otimes \mathrm{id}) \tau^\eta_{\mathcal{H}\mathcal{N}} \|_1,
\end{align*}
for $\eta = \frac{K-n+5}{K+n-5}$, provided that $n \geq N^*(K/(n-5))$.
\end{theorem}
One way to make sure that the input of the map is indeed restricted to states with less than $K$ photons is to replace $\Delta$ by $\Delta  \circ \cP^{\leq K}$.

In the following, for conciseness, we will denote by $\mathcal{H}$ the symmetric subspace:
\begin{align*}
\cH := F_{2,2,n}^{U(n)}.
\end{align*}

Let $\tau^\eta_{\mathcal{H}}$ be the normalized state corresponding to the projector $P_\eta$ defined in Eq.~\eqref{eqn:Peta}:
\begin{align*}
 \tau^\eta_{\mathcal{H}} = T(n,\eta)^{-1} \int_{ \mathcal{D}_\eta} |\Lambda,n\rangle \langle \Lambda,n| \d \mu_n(\Lambda)
\end{align*}
where 
\begin{align}
T(n,\eta) := \tr (P_\eta) = \frac{(n-1)(n-2)^2(n-3)\eta^4}{12(1-\eta)^4}. \label{eqn:defT}
\end{align}

Consider an orthonormal basis $\left\{Ê|\nu_i\rangle \right\}$ of $F_{2,2,n}^{U(n)}$ and define the non normalizable operator
\begin{align*}
|\Phi\rangle_{\cH \cN} := \sum_{i} |\nu_i\rangle_{\cH} |\nu_i\rangle_{\cN}.
\end{align*}

A conjecture for an explicit such orthonormal basis was given in \cite{lev16}, but we do not need to have such an explicit basis for our present purpose.

Let us further define the state $|\Phi^\eta\rangle \in F_{2,2,n}^{U(n)} \otimes F_{2,2,n}^{U(n)}$:
\begin{align*}
|\Phi^\eta\rangle = \left(\sqrt{ \tau^\eta} \otimes \1\right) |\Phi\rangle.
\end{align*}
It is well-known that $|\Phi^\eta\rangle$ is a purification of $\tau_{\mathcal{H}}^{\eta}$:
\begin{align*}
\tr_{\mathcal{N}} \left(|\Phi^\eta\rangle\langle \Phi^\eta|_{\mathcal{H}\mathcal{N}} \right)=\tau_{\mathcal{H}}^{\eta}.
\end{align*}

Recall that  $F_{2,2,n}^{U(n), \leq K}$ denotes the finite-dimensional subspace of  $F_{2,2,n}^{U(n)}$ with less than $K$ excitations.
\begin{lemma}\label{lem:measurement}
Let $\rho$ be an arbitrary density operator on $F_{2,2,n}^{U(n), \leq K}$. Then there exists a binary measurement $\mathcal{M} = \{M_\cN, \1_\cN-M_\cN\}$ on $\cN$ applied to $|\Phi^\eta\rangle \in \cH \otimes \cN$ that successfully prepares the state $\rho$ with probability at least $\frac{1}{2T(n,\eta)}$.
\end{lemma}

\begin{proof}
To avoid cluttering up the notations, let us write $\tau$ instead of $\tau^\eta_{\cH}$.
Recall that $\tau \geq p \1_{F_{2,2,n}^{U(n), \leq K}}$ with $p = \frac{1}{2T(n,\eta)}$, as a consequence of Corollary \ref{corol16}.

Let us define the non negative operator $M := p \tau^{-1/2} \rho \tau^{-1/2}$.  Since $p^{-1} \tau \geq \1$ on the support of $\rho \leq \1$, the operator $M$ satisfies
\begin{align*}
0 \leq M \leq \1.
\end{align*}
Let us define the measurement $\mathcal{M} = \{ M, \1-M\}$.
Performing this measurement on state $|\Phi^\eta\rangle$ prepares the state 
\begin{align*}
\tr_{\cN} \left( (1 \otimes M^{1/2}) |\Phi^\eta\rangle \langle \Phi^\eta | (1 \otimes M^{1/2}) \right)
\end{align*}
with probability $ \langle  \Phi^\eta | (1 \otimes M) |\Phi^\eta\rangle$.
This state can be written: 
\begin{align}
\tr_{\cN} \left( (1 \otimes M^{1/2}) |\Phi^\eta\rangle \langle \Phi^\eta | (1 \otimes M^{1/2})\right) &= \tr_{\cN} \left( (1 \otimes M^{1/2})\left(\sqrt{ \tau} \otimes \1\right) |\Phi\rangle \langle \Phi |   \left(\sqrt{ \tau} \otimes \1\right) (1 \otimes M^{1/2}) \right)  \nonumber \\
&= \tr_{\cN} \left( (\tau^{1/2} \otimes M^{1/2}) \sum_{i,j} |\nu_i \rangle \langle \nu_j| \otimes  |\nu_i \rangle \langle \nu_j| (\tau^{1/2} \otimes M^{1/2}) \right) \nonumber \\
&= \sum_{i,j} \tau^{1/2}  |\nu_i \rangle \langle \nu_j| \tau^{1/2}   \langle \nu_j| M^{1/2} M^{1/2}  |\nu_i \rangle  \nonumber\\
&= \sum_{i,j} \tau^{1/2}  |\nu_i \rangle \langle \nu_j| \tau^{1/2}   \langle \nu_i| M^{1/2} M^{1/2}  |\nu_j \rangle \label{eqn:inv} \\
&= \sum_{i,j} \tau^{1/2}  |\nu_i \rangle  \langle \nu_i| M^{1/2} M^{1/2}  |\nu_j \rangle  \langle \nu_j| \tau^{1/2}  \nonumber \\
&=  \tau^{1/2}  M^{1/2} M^{1/2}  \tau^{1/2} \nonumber \\
&=  \tau^{1/2}  p \tau^{-1/2} \rho \tau^{-1/2}  \tau^{1/2} \nonumber \\
&= p \rho, \nonumber
\end{align}
and it is obtained with probability $p$.
In Eq.~\eqref{eqn:inv}, we used that $M$ is symmetric, that is $\langle \lambda_i |M |\lambda_j\rangle = \langle \lambda_j |M |\lambda_i\rangle$.
\end{proof}

\begin{lemma} \label{lem:dim}
For $k\geq 0$ and $n\geq 4$, the dimensions of $V_{=K}$ and $V_{\leq K} = F_{2,2,n}^{U(n), \leq K}$ are given by
\begin{align*}
\mathrm{dim} \, V_{=K} = \tbinom{K+3}{3} \quad \text{and} \quad \mathrm{dim} \, V_{\leq K} = \tbinom{K+4}{4}.
\end{align*}
\end{lemma}

\begin{proof}
It was proven in \cite{lev16} that the vectors $(Z_{1,1})^i (Z_{1,2})^j (Z_{2,1})^k (Z_{2,2})^{\ell}$ are independent (provided than $n\geq 4$), which means that the dimension of $V_{=K}$ is the cardinality of the sets of quadruples $\{(i,j,k,\ell) \in \N^4 \: : \: i+j+k+\ell =K\}$. This number is $\tbinom{K+3}{3}$. More generally, the number of $t$-uples of nonnegative integers that sum to $K$ is $\tbinom{n+K-1}{n-1}$.
Since the subspaces $V_{=K}$ are orthogonal, it follows that $\mathrm{dim} \,  V_{\leq K} = \sum_{k=0}^K \mathrm{dim} \, V_{=k}$, which can be computed explicitly. Alternatively, one can see that the space $V_{\leq K}$ of quadruples $(i,j,k,\ell)$ summing to $K-m$ for some integer $m \leq K$ corresponds to the space of $5$-uples $(i,j,k,\ell, m)$ that sum to $K$. 
\end{proof}

\begin{lemma}\label{lem:design}
For any $K$ and $n$ integers, there exists a finite subset $\mathcal{U} \subset U(n)$, such that for any state $\rho$ with support on $F_{1,1,n}^{\leq K}$, the subspace of $F_{1,1,n}$ restricted to states with less than $K$ photons,  the following holds: 
\begin{align*}
\frac{1}{|\mathcal{U}|} \sum_{u \in \mathcal{U}} V_u \rho V_u^\dagger =  \int V_u \rho V_u^\dagger \d u,
\end{align*}
where $\d u$ is the normalized Haar measure on $U(n)$.
\end{lemma}

Note that by definition of the Haar measure, the state $ \int V_u \rho V_u^\dagger \d u$ is invariant under the application of any unitary $u' \in U(n)$: $V_{u'}  \int V_u \rho V_u^\dagger \d u V_{u'}^\dagger =  \int V_u \rho V_u^\dagger \d u$, which means that it has support on $F_{1,1,n}^{U(n), \leq K}$.

\begin{proof}
By linearity, it is sufficient to establish the lemma for pure states $|\psi\rangle \in F_{1,1,n}^{\leq K}$. 
Such a state can be written as
\begin{align*}
|\psi\rangle = \sum_{\substack{k_1, \ldots, k_n, \ell_1, \ldots \ell_n\\ \sum k_i + \ell_i \leq K}} \lambda_{k_1 \ldots k_n, \ell_1 \ldots \ell_n} \prod_{i=1}^n \left(a_i^\dagger\right)^{k_i} \left(b_i^\dagger\right)^{\ell_i} |0\rangle.
\end{align*}

Applying $V_u$ maps $a_i^\dagger$ to $\sum_{j=1}^n u_{i,j} a_j^\dagger$ and $b_i^\dagger$ to $\sum_{j=1}^n \overline{u}_{i,j} b_j^\dagger$.
In other words, the function $f: u \mapsto V_u |\psi\rangle \langle \psi | V_u^\dagger$ is a polynomial of degree at most $K$ in $u$ and $\overline{u}$.
Taking  $\mathcal{U}$ to be a $K$-design of $U(n)$, we obtain that
\begin{align*}
\frac{1}{|\mathcal{U}|} \sum_{u \in \mathcal{U}} f(u) = \int f(u) \d u,
\end{align*}
which proves the result. 
\end{proof}

We recall the following theorem that was established in \cite{lev16}.
\begin{theorem} \label{theo:purification}
Any density operator $\rho \in \mathfrak{S}(F_{1,1,n})$ invariant under $U(n)$ admits a purification in $F_{2,2,n}^{U(n)}$.
\end{theorem}

\begin{lemma}\label{lem:symmetrization}
It is sufficient to consider states $\rho_{\cH \mathcal{N}} $ with support on $F_{2,2,n}^{U(n),\leq K}$ when computing the diamond norm of Theorem \ref{thm:postselection}.
\end{lemma}

\begin{proof}
Consider a state $\rho_{\cH \cN}$ with support on $F_{2,2,n}^{\leq K}$. Let $\mathcal{U}$ be a finite set of unitaries as promised by Lemma \ref{lem:design}. Let $\{ |u\rangle_{\mathcal{C}}\}_{u \in \mathcal{U}}$ be an orthogonal basis for some classical register $\mathcal{C}$. The following sequence of equalities holds:
\begin{align}
\|(\Delta \otimes \1) \rho_{\cH \cN}\|_1 &= \frac{1}{|\mathcal{U}|} \sum_{u \in \mathcal{U}} \|( \Delta \otimes \1) (\rho_{\cH \cN} \otimes |u\rangle \langle u|_{\mathcal{C}}) \|_1 \nonumber \\
&=  \left\|\frac{1}{|\mathcal{U}|} \sum_{u \in \mathcal{U}} ( \Delta \otimes \1) (\rho_{\cH \cN} \otimes |u\rangle \langle u|_{\mathcal{C}}) \right\|_1  \label{eqn00} \\
&=  \left\|\frac{1}{|\mathcal{U}|} \sum_{u \in \mathcal{U}} ( \mathcal{K}_u \circ \Delta \otimes \1) (\rho_{\cH \cN} \otimes |u\rangle \langle u|_{\mathcal{C}}) \right\|_1 \label{eqn01} \\
&= \left\|\frac{1}{|\mathcal{U}|} \sum_{u \in \mathcal{U}} (  \Delta \circ u \otimes \1) (\rho_{\cH \cN} \otimes |u\rangle \langle u|_{\mathcal{C}}) \right\|_1 \label{eqn02} \\
&= \left\| (  \Delta \otimes \1) \left(\frac{1}{|\mathcal{U}|} \sum_{u \in \mathcal{U}}  ((u \circ \1) \rho_{\cH \cN} \otimes |u\rangle \langle u|_{\mathcal{C}})\right) \right\|_1 \nonumber 
\end{align}
where we used that the classical states $|u\rangle$ are all pairwise orthogonal in Eq.~\eqref{eqn00}, that $\cK_u$ is trace preserving in Eq.~\eqref{eqn01}, that $\cK_u \circ \Delta = \Delta \circ u$ in Eq.~\eqref{eqn02}.
Consider now the reduced state $\tilde{\rho}_{\cH}$:
\begin{align*}
\tilde{\rho}_{\cH} = \tr_{\cN \cC}\left(\frac{1}{|\mathcal{U}|} \sum_{u \in \mathcal{U}}  ((u \circ \1) \rho_{\cH \cN} \otimes |u\rangle \langle u|_{\mathcal{C}})\right) =\frac{1}{|\cU|} \sum_{u \in \mathcal{U}} V_u \rho_{\cH} V_u^\dagger = \int V_u \rho_{\cH} V_u^{\dagger} \d u
\end{align*}
where the last equality follows from Lemma \ref{lem:design}.
Theorem \ref{theo:purification} now assures the existence of some purification $\tilde{\rho}_{\cH \cN}$ of $\tilde{\rho}_{\cH}$ in $F_{2,2,n}^{U(n), \leq K}\cong \cH \otimes \cN$. 
In particular, there exists a CPTP map $g: \mathrm{End}(\cN) \to \mathrm{End}(\cN \otimes \cC)$ such that 
\begin{align*}
\frac{1}{|\mathcal{U}|} \sum_{u \in \mathcal{U}}  ((u \circ \1) \rho_{\cH \cN} \otimes |u\rangle \langle u|_{\mathcal{C}}) = (\1_\cH \otimes g) \tilde{\rho}_{\cH \cN}.
\end{align*}
Since $g$ is trace preserving, it further implies that 
\begin{align*}
\|(\Delta \otimes \1) \rho_{\cH \cN}\|_1 =\|(\Delta \otimes \1) (\1_\cH \otimes g) \overline{\rho}_{\cH \cN}\|_1  =\|(\Delta \otimes \1) \tilde{\rho}_{\cH \cN}\|_1,
\end{align*}
which concludes the proof
\end{proof}

We are now in position to prove Theorem \ref{thm:postselection}.
\begin{reptheorem}{thm:postselection}
Let $\Delta: \mathrm{End}(F_{1,1,n}^{\leq K}) \to  \mathrm{End}(\mathcal{H}')$ such that for all $u \in U(n)$, there exists a CPTP map $\mathcal{K}_u: \mathrm{End}(\mathcal{H}') \to \mathrm{End}(\mathcal{H}')$ such that $\Delta \circ u = \mathcal{K}_u: \circ \Delta$, then 
\begin{align*}
\|\Delta\|_\diamond \leq 2 T(n,\eta) \|(\Delta \otimes \mathrm{id}) \tau^\eta_{\mathcal{H}\mathcal{N}} \|_1,
\end{align*}
for $\eta = \frac{K-n+5}{K+n-5}$, provided that $n \geq N^*(K/(n-5))$.
\end{reptheorem}

\begin{proof}
According to Lemma \ref{lem:symmetrization}, it is sufficient to prove the theorem for a state $\rho_{\cH \cN}$ on $F_{2,2,n}^{U(n), \leq K}$.
Lemma \ref{lem:measurement} guarantees the existence of a trace-non-increasing map $\mathcal{T}$ from a copy of $F_{2,2,n}^{U(n), \leq K}$ to $\C$ such that 
\begin{align*}
\rho_{\cH \cN} = {2T(n,\eta)}(\1 \otimes \mathcal{T}) (|\Phi^\eta\rangle \langle \Phi^\eta |).
\end{align*}  
This gives 
\begin{align*}
(\Delta \otimes \1) \rho_{\cH \cN} = {2T(n,\eta)}(\Delta \otimes \mathcal{T}) (|\Phi^\eta\rangle \langle \Phi^\eta |)
\end{align*}  
and finally that
\begin{align*}
\|(\Delta \otimes \1) \rho_{\cH \cN}\|_1 ={2T(n,\eta)} \|(\Delta \otimes \mathcal{T}) (|\Phi^\eta\rangle \langle \Phi^\eta |)\|_1.
\end{align*}
\end{proof}


\section{Security against collective attacks provides a bound on $\| \cR \circ \Delta \circ \cP\|_{\diamond}$}
\label{sec:collective}

In order to exploit Theorem \ref{thm:postselection}, one needs an upper bound on $\|((\cR \circ \Delta \circ \cP^{\leq K})\otimes \mathrm{id}) \tau^\eta_{\mathcal{H}\mathcal{N}} \|_1$. We will see that such a bound can be obtained if the protocol is known to be secure against Gaussian collective attacks. 
For this, we follow the same strategy as in \cite{CKR09}.
Let us first recall the definition of being secure against Gaussian collective attacks. 

\begin{defn}
The QKD protocol $\cE_0$ is $\eps$-secure against Gaussian collective attacks if
\begin{align}
\| ((\cE_0-\cF_0)\otimes \mathrm{id})(|\Lambda,n\rangle \langle \Lambda,n|)\|_1\leq \eps \label{eqn:sec-coll}
\end{align}
for all $\Lambda \in \cD$.
\end{defn}

We show the following result. 
\begin{theorem}
With the previous notations, if $\cE_0$ is $\eps$-secure against Gaussian collective attacks, then 	
\begin{align*}
 \|(\cR \circ \Delta \circ \cP^{\leq K}  \otimes \mathrm{id}) \tau^\eta_{\mathcal{H}\mathcal{N}} \|_1 \leq \eps,
\end{align*}
where $\tau^\eta_{\mathcal{H}\mathcal{N}}$ is a purification of $\tau^\eta_{\mathcal{H}}$. Here $\cR$ is an additional privacy amplification step that reduces the key by $\lceil 2 \log_2 \tbinom{K+4}{4} \rceil$ bits and $\cP^{\leq K}$ is the projection onto $F_{1,1,n}^{\leq K}$.
\end{theorem}

\begin{proof}
Recall that 
\begin{align*}
 \tau^\eta_{\mathcal{H}} = T(n,\eta)^{-1} \int_{\mathcal{D}_\eta} |\Lambda,n\rangle \langle \Lambda,n| \d \mu_n(\Lambda)
\end{align*}
where 
\begin{align*}
T(n,\eta) := \tr (P_\eta) = \frac{(n-1)(n-2)^2(n-3)\eta^4}{12(1-\eta)^4}.
\end{align*}
By linearity, it holds that 
\begin{align*}
\| ((\cE-\cF)\otimes \mathrm{id})( \tau^\eta_{\mathcal{H}})\|_1  &= 
 \| ((\cE-\cF)\otimes \mathrm{id})( T(n,\eta)^{-1} \int_{\Lambda \in \mathcal{D}_\eta} |\Lambda,n\rangle \langle \Lambda,n| \d \mu_n(\Lambda) )\|_1     \\
  &\leq  \eps T(n,\eta)^{-1}\left\| \int_{\Lambda \in  \mathcal{D}_\eta}  \d \mu_n(\Lambda) \right\|_1     \\  &= \eps
\end{align*}
In order to obtain the theorem, we need to consider a purification of $\tau_{\cH \cN}^{\eta}$. Since $\cP^{\leq K}$ restricts the states to live in a space of dimension at most $\mathrm{dim} \, F_{2,2,n}^{U(n), \leq K} = \tbinom{K+4}{4}$ (according to Lemma \ref{lem:dim}), it implies that the purifying system $\cN$ can be chosen of this dimension. 
Giving this extra system to Eve can at most provide her with a limited amount of information. Applying an additional privacy amplification step $\cR$ ensures that the protocol remains $\eps$-secure for the state $\tau_{\cH \cN}^\eta$ thanks to the leftover hashing lemma (Theorem 5.1.1 of \cite{Ren08}):
\begin{align*}
 \|(\cR \circ \Delta \circ \cP^{\leq K}  \otimes \mathrm{id}) \tau^\eta_{\mathcal{H}\mathcal{N}} \|_1 \leq  \|( \Delta \circ \cP^{\leq K}  \otimes \mathrm{id}) \tau^\eta_{\mathcal{H}} \|_1.
\end{align*}

\end{proof}

Combining this result with Theorem \ref{thm:postselection} yields Theorem \ref{thm:diamond-protocol}. 
\begin{reptheorem}{thm:diamond-protocol}
With the previous notations, if $\cE_0$ is $\eps$-secure against Gaussian collective attacks, then 	
\begin{align*}
 \|\cR \circ (\cE_0-\cF_0) \circ \cP^{\leq K}  \|_{\diamond} \leq 2 T(n,\eta)  \eps.
\end{align*}
\end{reptheorem}


\section{Energy test}
\label{sec:test}

The goal of this section is to prove the following result.

\begin{reptheorem}{thm:test}
For integers $n,k \geq 1$, and $d_A, d_B >0$, define $K = n(d'_A + d'_B)$ for $d'_{A/B} = d_{A/B} g(n,k,\eps/4)$ for the function $g$ defined in Eq.~\eqref{eqn:g}. Then
\begin{align*}
\big\| \big(\1 - \cP(n,K)\big) \circ \cT(k, d_A, d_B)\big\|_{\diamond} \leq \eps.
\end{align*}
\end{reptheorem}

For $d >0$, let us introduce the following operators on $\cH^{\otimes n}$ for a single-mode Fock space $\cH$:
\begin{align*}
T_n^{d} &:= \frac{1}{\pi^n} \int_{\sum_{i=1}^n  |\alpha_i|^2 \geq n d} |\alpha_1\rangle\langle \alpha_1| \otimes \ldots \otimes |\alpha_n\rangle \langle \alpha_n| \d \alpha_1 \ldots \alpha_n\\
U_n^{d} &:= \sum_{m = n d+1}^\infty \Pi_{m}^n,
\end{align*}
where $\Pi_m^n$ is the projector onto the subspace of $\cH^{\otimes n}$ spanned by Fock states containing $m$ photons:
\begin{align*}
\Pi_m^n = \sum_{m_1+\ldots+m_n=m} |m_1, \ldots, m_n\rangle \langle m_1, \ldots, m_n|.
\end{align*}
In words, $T_n^{d}$ is the sum of the projectors onto products of coherent states such that the total squared amplitude is greater than $n d$ and $U_n^{d}$ is the projector onto Fock states containing more that $n d$ photons. Intuitively, both operators should be ``close'' to each other. This is formalized with the following lemma that was proven in \cite{LGRC13}.
\begin{lemma}\label{lem:LGRC}
For any integer $n$ and any $d\geq0$, it holds that
\begin{align*}
U_n^{d} \leq 2 T_n^{d}.
\end{align*}
\end{lemma}

The following lemma results from the definitions of $U_n^d$ and $\cP^{\leq K}$, the projector onto $F_{1,1,n}^{\leq K}$.
\begin{lemma}\label{lem:obs}
For any $d_A, d_B \geq 0$ and integer $K$ such that $K \leq n(d_A+d_B)$, it holds that
\begin{align*}
\1_{\cH_A^{\otimes n} \otimes \cH_B^{\otimes n}} - \cP^{\leq K} \leq U_n^{d_A}\otimes  \1_{\cH_B^{\otimes n}} + \1_{\cH_A^{\otimes n}} \otimes U_n^{d_B}.
\end{align*}
\end{lemma}
\begin{proof}
The left hand side is the projector onto the states of $\cH_A^{\otimes n} \otimes \cH_B^{\otimes n}$ containing strictly more than $K$ photons. Any such state must contain either at least $n d_A$ photons in $\cH_A^{\otimes n}$ or at least $K - n d_A$ photons in $\cH_B^{\otimes n}$, for any possible value of $d_A$. 
This proves the claim.
\end{proof}

Combining Lemmas \ref{lem:LGRC} and \ref{lem:obs}, we obtain the immediate corollary.
\begin{corol}\label{cor:proj}
For any $d_A, d_B \geq 0$ and integer $K$ such that $K \leq n(d_A+d_B)$, it holds that
\begin{align*}
\1_{\cH_A^{\otimes n} \otimes \cH_B^{\otimes n}} - \cP^{\leq K} \leq 2 T_n^{d_A}\otimes  \1_{\cH_B^{\otimes n}} + 2\1_{\cH_A^{\otimes n}} \otimes T_n^{d_B}.
\end{align*}
\end{corol}

Recall that the heterodyne measurement corresponds to a projection onto (Glauber) coherent states, and is described by the resolution of the identity:
\begin{align*}
\1_{\cH^{\otimes k}} = \frac{1}{\pi^k} \int_{\C^k} |\alpha_1\rangle \langle \alpha_1| \otimes \ldots \otimes |\alpha_k\rangle \langle \alpha_k| \d\alpha_1 \ldots \d \alpha_k.
\end{align*}
In other words, measuring a state $\rho$ on $\cH^{\otimes k}$ with heterodyne detection outputs the result $(\alpha_1, \ldots, \alpha_k) \in \C^k$ with probability
\begin{align*}
\mathrm{Pr}_\rho(\alpha_1, \ldots, \alpha_k) =\frac{1}{\pi^k} \tr( \rho  |\alpha_1\rangle \langle \alpha_1| \otimes \ldots \otimes |\alpha_k\rangle \langle \alpha_k|).
\end{align*}

Laurent and Massart \cite{LM00} established the following tail bounds for $\chi^2(D)$ distributions. 
\begin{lemma}[Laurent and Massard \cite{LM00}] \label{lem:LM}
Let $U$ be a $\chi^2$ statistic with $D$ degrees of freedom. For any $x >0$, 
\begin{align*}
\mathrm{Pr}[U-D \geq 2\sqrt{D x} + 2 Dx] \leq \exp(-x) \quad \text{and} \quad \mathrm{Pr}[D-U \geq 2\sqrt{Dx}] \leq \exp(-x). 
\end{align*}
\end{lemma}

\begin{defn}
A state $\rho$ on $\cH^{\otimes n} = F(\C^n)$ is said \emph{rotationally invariant} is $V_u \rho V_u^\dagger = \rho$ for all $u \in U(n)$.
\end{defn}
In particular, the state $\int V_u \rho V_U^\dagger \d u$ is invariant if $\d u$ is the Haar measure on $U(n)$.

\begin{lemma}\label{lem:36}
Let $\rho$ be an rotationally invariant state on $\cH^{\otimes (n+k)}$. Then, for any $d >0$, 
\begin{align*}
\tr \left[ ( T_n^{d'} \otimes (\1-T_k^d)) \rho \right] \leq \eps,
\end{align*}
for $d' = g(n,k,\eps) d$ and 
\begin{align}
 g(n,k,\eps) = \frac{1 + 2 \sqrt{\frac{\ln (2/\eps)}{2n}} +  \frac{\ln (2/\eps)}{n}}{1-2{\sqrt{\frac{\ln (2/\eps)}{2k}}}}. \label{eqn:g}
\end{align}
\end{lemma}

\begin{proof}
By definition, $\tr[T_k^d \rho]$ is the probability that the outcome $(\alpha_1, \ldots, \alpha_k)\in \C^k$ obtained by measuring the last $k$ modes of the state $\rho$ with heterodyne detection satisfies $\sum_{i=1}^k |\alpha_i|^2 \geq k d$.
Similarly, $\tr \left[ ((T_n^{d'} \otimes (\1-T_k^d)) \rho \right]$ is the probability that the outcome of measuring the $n+k$ modes of $\rho$ with heterodyne detection yields a vector $(\alpha_1, \ldots, \alpha_{n+k})$ such that 
\begin{align*}
Y_n := \sum_{i=1}^n |\alpha_i|^2 \geq  n d' \quad \text{and} \quad Y_k :=\sum_{i=1}^k |\alpha_{n+i}|^2 \leq k d.
\end{align*}
Since the state is rotationally invariant, it means that the random vector $(\alpha_1, \ldots, \alpha_{n+k})$ is uniformly distributed on the sphere of radius $M$  in $\C^{n+k}$, conditioned on the fact that the modulus is $\sqrt{\sum_{i=1}^{n+k} |\alpha_i|^2}=M$. Equivalently, one can consider the $2(n+k)$-dimensional real vector $(\fR(\alpha_1), \fI(\alpha_1), \ldots, \fR(\alpha_1), \fI(\alpha_1))$ which is uniformly distributed over the sphere in $\R^{2(n+k)}$. Here $\fR(\alpha_1)$ and $\fI(\alpha)$ refer respectively to the real and imaginary part of $\alpha$. 
We obtain
\begin{align*}
\tr \left[ ( T_n^{d'} \otimes (\1-T_k^d)) \rho \right] &= \mathrm{Pr}[(Y_n \geq  nd' ) \wedge (Y_k \leq kd)]\\
& \leq \mathrm{Pr}[kd Y_n \geq nd' Y_k]
\end{align*}
where the inequality is a simple consequence of the fact that the rectangle $[nd',\infty] \times [0, kd]$ is a subset of the triangle $\{ (x,y) \in [0,\infty]^2 \: : \: kd x \geq nd' y\}$.

It is well-known that the uniform distribution over the unit sphere of $\R^{2(n+k)}$ can be generated by sampling $2(n+k)$ normal variables with 0 mean and unit variance. 
In that case, the squared norm $\sum_{i=1}^n |\alpha_i|^2$ is simply a $\chi^2$ variable with $2n$ degrees of freedom while $\sum_{i=1}^k |\alpha_{n+i}|^2$ corresponds to an independent $\chi^2$ variable with $2k$ degrees of freedom. Let us denote by $Z_n$ and $Z_k$ the corresponding random variables: $Z_n \sim \chi^2(2n)$, $Z_k \sim \chi^2(2k)$. 
Since $(Y_n, Y_k)$ and $(Z_n, Z_k)$ follow the same distribution, up to rescaling, we obtain that 
\begin{align*}
\mathrm{Pr}[kd Y_n \geq nd' Y_k] = \mathrm{Pr}[kd Z_n \geq nd' Z_k].
\end{align*}
This is particularly useful because it means that there is no need to enforce normalization explicitly.
Finally, using now that the triangle $\{ (x,y) \in [0,\infty]^2 \: : \: kd x \geq dd' y\}$ is a subset of the union of the rectangles $[\alpha nd',\infty]\times [0,\infty]$ and $[0,\infty] \times [0,\alpha kd]$ for any $\alpha >0$, it follows that 
\begin{align*}
\mathrm{Pr}[kd Z_n \geq nd' Z_k] \leq  \mathrm{Pr}[ Z_n \geq  \alpha n d'] + \mathrm{Pr}[ Z_k \leq \alpha k d].
\end{align*}
Choosing $\alpha$ such that
\begin{align*}
\alpha k d = 2k \left(1 - 2 \sqrt{\frac{\ln(\eps/2)}{2k}}\right)
\end{align*}
and applying the lower bounds on the tails of the $\chi^2$ distribution given in Lemma \ref{lem:LM} gives
\begin{align*}
\mathrm{Pr}[ Z_n \geq  \alpha n d'] \leq \frac{\eps}{2}, \quad \mathrm{Pr}[ Z_k \leq \alpha k d] \leq \frac{\eps}{2}.
\end{align*}
This establishes that 
\begin{align*}
\tr \left[ ( T_n^{d'} \otimes (\1-T_k^d)) \rho \right] \leq \eps, 
\end{align*}
which concludes the proof.
\end{proof}

We are now ready to define and analyze the energy test.
Alice and Bob perform a random rotation of their data according to a unitary $u \in U(n)$ chosen from the Haar measure on $U(n)$, and measure the last $k$ modes of their respective state with heterodyne detection. They compute the squared norm of their respective vectors and obtain two values $Y_A$ for Alice and $Y_B$ for Bob. The test depends on three parameters: the number $k$ of modes which are measured, a maximum value for Alice $d_A$ and a maximum value for Bob, $d_B$. 
The test $\mathcal{T}(k, d_A, d_B)$ passes if 
\begin{align*}
Y_A \leq k d_A \quad \text{and} \quad Y_B \leq k d_B.
\end{align*} 

We are interested in the probability of passing the test and failing for the remaining modes to contain less than $K$ photons, more precisely in the quantity
\begin{align*}
\|(\1 - \cP) \circ \cT\|_{\diamond}.
\end{align*}

Let us denote by $\mathrm{Inv}( \fS(\cH^{\otimes (n+k)}))$ the set of density matrices which are invariant under the action of $U(n+k)$.

\begin{reptheorem}{thm:test}
For integers $n,k \geq 1$, and $d_A, d_B >0$, define $K = n(d'_A + d'_B)$ for $d'_{A/B} = d_{A/B} g(n,k,\eps/4)$ for the function $g$ defined in Eq.~\eqref{eqn:g}. Then
\begin{align*}
\big\| \big(\1 - \cP(n,K)\big) \circ \cT(k, d_A, d_B)\big\|_{\diamond} \leq \eps.
\end{align*}
\end{reptheorem}

\begin{proof}
Writing $\cP$ and $\cT$ for conciseness, the definition of the diamond norm yields:
\begin{align}
\|(\1 - \cP) \circ \cT\|_{\diamond} &= \max_{\rho \in \cH_{AB}^{\otimes (n+k)} \otimes  \cH_{AB}^{\otimes (n+k)}} \big\| \big((\1 - \cP) \circ \cT)\otimes \1_{\cH_{AB}^{\otimes (n+k)}}\big) (\rho)\big\|_{1} \nonumber\\
&= \max_{\rho \in \fS(\cH_{AB}^{\otimes (n+k)})} \|(\1 - \cP) \circ \cT (\rho)\|_{1} \label{eqn:nonneg}\\
& \leq \max_{\rho \in \mathrm{Inv}\left( \fS(\cH_{AB}^{\otimes (n+k)}\right)} \|(\1 - \cP) \circ \cT (\rho)\|_{1} \label{eqn:invar}\\
& \leq \max_{\rho \in \mathrm{Inv}\left( \fS(\cH_{AB}^{\otimes (n+k)}\right)} \| (U_n^{d'_A} \otimes \1 +\1 \otimes U_n^{d'_B}) \circ \big( (\1- T_k^{d_A} )\otimes (\1-T_k^{d_B}) \big)(\rho)\|_{1} \label{eqn:sum} \\
& = \max_{\rho \in \mathrm{Inv}\left( \fS(\cH_{AB}^{\otimes (n+k)}\right)} \| (U_n^{d'_A} \circ (\1-T_k^{d_A}) + U_n^{d'_B} \circ (\1- T_k^{d_B})) (\rho)\|_{1} \nonumber \\
& \leq \max_{\rho \in \mathrm{Inv}\left( \fS(\cH_{A}^{\otimes (n+k)}\right)} \| (U_n^{d'_A}  \circ (\1-T_k^{d_A} )) (\rho)\|_{1}  + \max_{\rho \in \mathrm{Inv}\left( \fS(\cH_{B}^{\otimes (n+k)}\right)} \| ( U_n^{d'_B} \circ  (\1-T_k^{d_B})) (\rho)\|_{1} \label{eqn:triang}\\
& \leq 2 \max_{\rho \in \mathrm{Inv}\left( \fS(\cH_{A}^{\otimes (n+k)}\right)} \| (T_n^{d'_A}  \circ (\1-T_k^{d_A}) ) (\rho)\|_{1}   +2 \max_{\rho \in \mathrm{Inv}\left( \fS(\cH_{B}^{\otimes (n+k)}\right)} \| ( T_n^{d'_B} \circ  (\1-T_k^{d_B})) (\rho)\|_{1}\\
&\leq  \eps \label{eqn:final}
\end{align}
where we used that  $\big((\1 - \cP) \circ \cT)\otimes \1_{\cH_{AB}^{\otimes (n+k)}}\big) (\rho)$ is a nonnegative operator in Eq.~\eqref{eqn:nonneg}, the fact that both $\cP$ and $\cT$ are rotationally invariant in Eq.~\eqref{eqn:invar}, Lemma \ref{lem:obs} in Eq.~\eqref{eqn:sum}, the triangle inequality in Eq.~\eqref{eqn:triang}, Lemma \ref{lem:36} in Eq.~\eqref{eqn:final}.
\end{proof}

\end{widetext}

\begin{thebibliography}{30}
\expandafter\ifx\csname natexlab\endcsname\relax\def\natexlab#1{#1}\fi
\expandafter\ifx\csname bibnamefont\endcsname\relax
  \def\bibnamefont#1{#1}\fi
\expandafter\ifx\csname bibfnamefont\endcsname\relax
  \def\bibfnamefont#1{#1}\fi
\expandafter\ifx\csname citenamefont\endcsname\relax
  \def\citenamefont#1{#1}\fi
\expandafter\ifx\csname url\endcsname\relax
  \def\url#1{\texttt{#1}}\fi
\expandafter\ifx\csname urlprefix\endcsname\relax\def\urlprefix{URL }\fi
\providecommand{\bibinfo}[2]{#2}
\providecommand{\eprint}[2][]{\url{#2}}

\bibitem[{\citenamefont{Portmann and Renner}(2014)}]{PR14}
\bibinfo{author}{\bibfnamefont{C.}~\bibnamefont{Portmann}} \bibnamefont{and}
  \bibinfo{author}{\bibfnamefont{R.}~\bibnamefont{Renner}},
  \emph{\bibinfo{title}{{Cryptographic Security of Quantum Key Distribution}}}
  (\bibinfo{year}{2014}), \eprint{arXiv preprint 1409.3525}.

\bibitem[{\citenamefont{Watrous}(2016)}]{wat16}
\bibinfo{author}{\bibfnamefont{J.}~\bibnamefont{Watrous}},
  \emph{\bibinfo{title}{Theory of quantum information}} (\bibinfo{year}{2016}).

\bibitem[{\citenamefont{Bennett and Brassard}(1984)}]{BB84}
\bibinfo{author}{\bibfnamefont{C.}~\bibnamefont{Bennett}} \bibnamefont{and}
  \bibinfo{author}{\bibfnamefont{G.}~\bibnamefont{Brassard}}, in
  \emph{\bibinfo{booktitle}{Proceedings of IEEE International Conference on
  Computers, Systems and Signal Processing}} (\bibinfo{year}{1984}), vol.
  \bibinfo{volume}{175}.

\bibitem[{\citenamefont{Christandl et~al.}(2009)\citenamefont{Christandl,
  K\"{o}nig, and Renner}}]{CKR09}
\bibinfo{author}{\bibfnamefont{M.}~\bibnamefont{Christandl}},
  \bibinfo{author}{\bibfnamefont{R.}~\bibnamefont{K\"{o}nig}},
  \bibnamefont{and} \bibinfo{author}{\bibfnamefont{R.}~\bibnamefont{Renner}},
  \bibinfo{journal}{Phys. Rev. Lett.} \textbf{\bibinfo{volume}{102}},
  \bibinfo{eid}{020504} (\bibinfo{year}{2009}).

\bibitem[{\citenamefont{Sheridan et~al.}(2010)\citenamefont{Sheridan, Le, and
  Scarani}}]{SLS10}
\bibinfo{author}{\bibfnamefont{L.}~\bibnamefont{Sheridan}},
  \bibinfo{author}{\bibfnamefont{T.}~\bibnamefont{Le}}, \bibnamefont{and}
  \bibinfo{author}{\bibfnamefont{V.}~\bibnamefont{Scarani}},
  \bibinfo{journal}{New J. Phys.} \textbf{\bibinfo{volume}{12}},
  \bibinfo{pages}{123019} (\bibinfo{year}{2010}).

\bibitem[{\citenamefont{Sheridan and Scarani}(2010)}]{SS10}
\bibinfo{author}{\bibfnamefont{L.}~\bibnamefont{Sheridan}} \bibnamefont{and}
  \bibinfo{author}{\bibfnamefont{V.}~\bibnamefont{Scarani}},
  \bibinfo{journal}{Phys. Rev. A} \textbf{\bibinfo{volume}{82}},
  \bibinfo{pages}{030301} (\bibinfo{year}{2010}).

\bibitem[{\citenamefont{Weedbrook et~al.}(2012)\citenamefont{Weedbrook,
  Pirandola, Garc{\'i}a-Patr\'on, Cerf, Ralph, Shapiro, and Lloyd}}]{WPG12}
\bibinfo{author}{\bibfnamefont{C.}~\bibnamefont{Weedbrook}},
  \bibinfo{author}{\bibfnamefont{S.}~\bibnamefont{Pirandola}},
  \bibinfo{author}{\bibfnamefont{R.}~\bibnamefont{Garc{\'i}a-Patr\'on}},
  \bibinfo{author}{\bibfnamefont{N.~J.} \bibnamefont{Cerf}},
  \bibinfo{author}{\bibfnamefont{T.~C.} \bibnamefont{Ralph}},
  \bibinfo{author}{\bibfnamefont{J.~H.} \bibnamefont{Shapiro}},
  \bibnamefont{and} \bibinfo{author}{\bibfnamefont{S.}~\bibnamefont{Lloyd}},
  \bibinfo{journal}{Rev. Mod. Phys.} \textbf{\bibinfo{volume}{84}},
  \bibinfo{pages}{621} (\bibinfo{year}{2012}).

\bibitem[{\citenamefont{Diamanti and Leverrier}(2015)}]{DL15}
\bibinfo{author}{\bibfnamefont{E.}~\bibnamefont{Diamanti}} \bibnamefont{and}
  \bibinfo{author}{\bibfnamefont{A.}~\bibnamefont{Leverrier}},
  \bibinfo{journal}{Entropy} \textbf{\bibinfo{volume}{17}},
  \bibinfo{pages}{6072} (\bibinfo{year}{2015}).

\bibitem[{\citenamefont{Weedbrook et~al.}(2004)\citenamefont{Weedbrook, Lance,
  Bowen, Symul, Ralph, and Lam}}]{WLB04}
\bibinfo{author}{\bibfnamefont{C.}~\bibnamefont{Weedbrook}},
  \bibinfo{author}{\bibfnamefont{A.~M.} \bibnamefont{Lance}},
  \bibinfo{author}{\bibfnamefont{W.~P.} \bibnamefont{Bowen}},
  \bibinfo{author}{\bibfnamefont{T.}~\bibnamefont{Symul}},
  \bibinfo{author}{\bibfnamefont{T.~C.} \bibnamefont{Ralph}}, \bibnamefont{and}
  \bibinfo{author}{\bibfnamefont{P.~K.} \bibnamefont{Lam}},
  \bibinfo{journal}{Phys. Rev. Lett.} \textbf{\bibinfo{volume}{93}},
  \bibinfo{pages}{170504} (\bibinfo{year}{2004}).

\bibitem[{\citenamefont{Leverrier}(2015)}]{Lev15}
\bibinfo{author}{\bibfnamefont{A.}~\bibnamefont{Leverrier}},
  \bibinfo{journal}{Phys. Rev. Lett.} \textbf{\bibinfo{volume}{114}},
  \bibinfo{pages}{070501} (\bibinfo{year}{2015}).

\bibitem[{\citenamefont{Leverrier et~al.}(2013)\citenamefont{Leverrier,
  Garc\'ia-Patr\'on, Renner, and Cerf}}]{LGRC13}
\bibinfo{author}{\bibfnamefont{A.}~\bibnamefont{Leverrier}},
  \bibinfo{author}{\bibfnamefont{R.}~\bibnamefont{Garc\'ia-Patr\'on}},
  \bibinfo{author}{\bibfnamefont{R.}~\bibnamefont{Renner}}, \bibnamefont{and}
  \bibinfo{author}{\bibfnamefont{N.~J.} \bibnamefont{Cerf}},
  \bibinfo{journal}{Phys. Rev. Lett.} \textbf{\bibinfo{volume}{110}},
  \bibinfo{pages}{030502} (\bibinfo{year}{2013}).

\bibitem[{\citenamefont{Renner}(2008)}]{Ren08}
\bibinfo{author}{\bibfnamefont{R.}~\bibnamefont{Renner}},
  \bibinfo{journal}{International Journal of Quantum Information}
  \textbf{\bibinfo{volume}{6}}, \bibinfo{pages}{1} (\bibinfo{year}{2008}).

\bibitem[{\citenamefont{Renner and Cirac}(2009)}]{RC09}
\bibinfo{author}{\bibfnamefont{R.}~\bibnamefont{Renner}} \bibnamefont{and}
  \bibinfo{author}{\bibfnamefont{J.~I.} \bibnamefont{Cirac}},
  \bibinfo{journal}{Phys. Rev. Lett.} \textbf{\bibinfo{volume}{102}},
  \bibinfo{pages}{110504} (\bibinfo{year}{2009}).

\bibitem[{\citenamefont{Cerf et~al.}(2001)\citenamefont{Cerf, Levy, and
  Van~Assche}}]{CLV01}
\bibinfo{author}{\bibfnamefont{N.~J.} \bibnamefont{Cerf}},
  \bibinfo{author}{\bibfnamefont{M.}~\bibnamefont{Levy}}, \bibnamefont{and}
  \bibinfo{author}{\bibfnamefont{G.}~\bibnamefont{Van~Assche}},
  \bibinfo{journal}{Phys. Rev. A} \textbf{\bibinfo{volume}{63}},
  \bibinfo{pages}{052311} (\bibinfo{year}{2001}).

\bibitem[{\citenamefont{Furrer et~al.}(2012)\citenamefont{Furrer, Franz, Berta,
  Leverrier, Scholz, Tomamichel, and Werner}}]{FBB12}
\bibinfo{author}{\bibfnamefont{F.}~\bibnamefont{Furrer}},
  \bibinfo{author}{\bibfnamefont{T.}~\bibnamefont{Franz}},
  \bibinfo{author}{\bibfnamefont{M.}~\bibnamefont{Berta}},
  \bibinfo{author}{\bibfnamefont{A.}~\bibnamefont{Leverrier}},
  \bibinfo{author}{\bibfnamefont{V.~B.} \bibnamefont{Scholz}},
  \bibinfo{author}{\bibfnamefont{M.}~\bibnamefont{Tomamichel}},
  \bibnamefont{and} \bibinfo{author}{\bibfnamefont{R.~F.}
  \bibnamefont{Werner}}, \bibinfo{journal}{{Phys. Rev. Lett.}}
  \textbf{\bibinfo{volume}{109}}, \bibinfo{pages}{100502}
  (\bibinfo{year}{2012}).

\bibitem[{\citenamefont{Leverrier et~al.}(2009)\citenamefont{Leverrier, Karpov,
  Grangier, and Cerf}}]{LKG09}
\bibinfo{author}{\bibfnamefont{A.}~\bibnamefont{Leverrier}},
  \bibinfo{author}{\bibfnamefont{E.}~\bibnamefont{Karpov}},
  \bibinfo{author}{\bibfnamefont{P.}~\bibnamefont{Grangier}}, \bibnamefont{and}
  \bibinfo{author}{\bibfnamefont{N.}~\bibnamefont{Cerf}}, \bibinfo{journal}{New
  J. Phys.} \textbf{\bibinfo{volume}{11}}, \bibinfo{pages}{115009}
  (\bibinfo{year}{2009}).

\bibitem[{\citenamefont{Furrer}(2014)}]{fur14}
\bibinfo{author}{\bibfnamefont{F.}~\bibnamefont{Furrer}},
  \bibinfo{journal}{Phys. Rev. A} \textbf{\bibinfo{volume}{90}},
  \bibinfo{pages}{042325} (\bibinfo{year}{2014}).

\bibitem[{\citenamefont{Pirandola et~al.}(2015)\citenamefont{Pirandola,
  Ottaviani, Spedalieri, Weedbrook, Braunstein, Lloyd, Gehring, Jacobsen, and
  Andersen}}]{POS15}
\bibinfo{author}{\bibfnamefont{S.}~\bibnamefont{Pirandola}},
  \bibinfo{author}{\bibfnamefont{C.}~\bibnamefont{Ottaviani}},
  \bibinfo{author}{\bibfnamefont{G.}~\bibnamefont{Spedalieri}},
  \bibinfo{author}{\bibfnamefont{C.}~\bibnamefont{Weedbrook}},
  \bibinfo{author}{\bibfnamefont{S.~L.} \bibnamefont{Braunstein}},
  \bibinfo{author}{\bibfnamefont{S.}~\bibnamefont{Lloyd}},
  \bibinfo{author}{\bibfnamefont{T.}~\bibnamefont{Gehring}},
  \bibinfo{author}{\bibfnamefont{C.~S.} \bibnamefont{Jacobsen}},
  \bibnamefont{and} \bibinfo{author}{\bibfnamefont{U.~L.}
  \bibnamefont{Andersen}}, \bibinfo{journal}{Nature Photonics}
  \textbf{\bibinfo{volume}{9}}, \bibinfo{pages}{397} (\bibinfo{year}{2015}).

\bibitem[{\citenamefont{Leverrier}(2016)}]{lev16}
\bibinfo{author}{\bibfnamefont{A.}~\bibnamefont{Leverrier}},
  \bibinfo{journal}{arXiv preprint 1612.05080}  (\bibinfo{year}{2016}).

\bibitem[{\citenamefont{Perelomov}(1972)}]{per72}
\bibinfo{author}{\bibfnamefont{A.}~\bibnamefont{Perelomov}},
  \bibinfo{journal}{Communications in Mathematical Physics}
  \textbf{\bibinfo{volume}{26}}, \bibinfo{pages}{222} (\bibinfo{year}{1972}).

\bibitem[{\citenamefont{Perelomov}(1986)}]{per86}
\bibinfo{author}{\bibfnamefont{A.}~\bibnamefont{Perelomov}},
  \emph{\bibinfo{title}{Generalized coherent states and their applications}}
  (\bibinfo{publisher}{Springer}, \bibinfo{year}{1986}).

\bibitem[{\citenamefont{Christandl et~al.}(2007)\citenamefont{Christandl,
  K{\"o}nig, Mitchison, and Renner}}]{CKMR07}
\bibinfo{author}{\bibfnamefont{M.}~\bibnamefont{Christandl}},
  \bibinfo{author}{\bibfnamefont{R.}~\bibnamefont{K{\"o}nig}},
  \bibinfo{author}{\bibfnamefont{G.}~\bibnamefont{Mitchison}},
  \bibnamefont{and} \bibinfo{author}{\bibfnamefont{R.}~\bibnamefont{Renner}},
  \bibinfo{journal}{Communications in Mathematical Physics}
  \textbf{\bibinfo{volume}{273}}, \bibinfo{pages}{473} (\bibinfo{year}{2007}).

\bibitem[{\citenamefont{Tomamichel and Renner}(2011)}]{TR11}
\bibinfo{author}{\bibfnamefont{M.}~\bibnamefont{Tomamichel}} \bibnamefont{and}
  \bibinfo{author}{\bibfnamefont{R.}~\bibnamefont{Renner}},
  \bibinfo{journal}{{Phys. Rev. Lett.}} \textbf{\bibinfo{volume}{106}},
  \bibinfo{pages}{110506} (\bibinfo{year}{2011}).

\bibitem[{\citenamefont{Tomamichel et~al.}(2012)\citenamefont{Tomamichel, Lim,
  Gisin, and Renner}}]{TLG12}
\bibinfo{author}{\bibfnamefont{M.}~\bibnamefont{Tomamichel}},
  \bibinfo{author}{\bibfnamefont{C.}~\bibnamefont{Lim}},
  \bibinfo{author}{\bibfnamefont{N.}~\bibnamefont{Gisin}}, \bibnamefont{and}
  \bibinfo{author}{\bibfnamefont{R.}~\bibnamefont{Renner}},
  \bibinfo{journal}{Nat. Comm.} \textbf{\bibinfo{volume}{3}},
  \bibinfo{pages}{634} (\bibinfo{year}{2012}).

\bibitem[{\citenamefont{Dupuis et~al.}(2016)\citenamefont{Dupuis, Fawzi, and
  Renner}}]{DFR16}
\bibinfo{author}{\bibfnamefont{F.}~\bibnamefont{Dupuis}},
  \bibinfo{author}{\bibfnamefont{O.}~\bibnamefont{Fawzi}}, \bibnamefont{and}
  \bibinfo{author}{\bibfnamefont{R.}~\bibnamefont{Renner}},
  \bibinfo{journal}{arXiv preprint 1607.01796}  (\bibinfo{year}{2016}).

\bibitem[{\citenamefont{Harrow}(2013)}]{har13}
\bibinfo{author}{\bibfnamefont{A.~W.} \bibnamefont{Harrow}},
  \bibinfo{journal}{arXiv preprint 1308.6595}  (\bibinfo{year}{2013}).

\bibitem[{\citenamefont{Helgason}(1979)}]{hel79}
\bibinfo{author}{\bibfnamefont{S.}~\bibnamefont{Helgason}},
  \emph{\bibinfo{title}{{Differential geometry, Lie groups, and symmetric
  spaces}}}, vol.~\bibinfo{volume}{80} (\bibinfo{publisher}{Academic press},
  \bibinfo{year}{1979}).

\bibitem[{\citenamefont{Muirhead}(1982)}]{mui82}
\bibinfo{author}{\bibfnamefont{R.~J.} \bibnamefont{Muirhead}},
  \emph{\bibinfo{title}{Aspects of multivariate statistical theory}}, Wiley
  series in probability and mathematical statistics. Probability and
  mathematical statistics (\bibinfo{publisher}{John Wiley \& Sons},
  \bibinfo{address}{New York}, \bibinfo{year}{1982}), ISBN
  \bibinfo{isbn}{0-471-09442-0}.

\bibitem[{\citenamefont{M{\"u}ller-Quade and Renner}(2009)}]{MKR09}
\bibinfo{author}{\bibfnamefont{J.}~\bibnamefont{M{\"u}ller-Quade}}
  \bibnamefont{and} \bibinfo{author}{\bibfnamefont{R.}~\bibnamefont{Renner}},
  \bibinfo{journal}{New J. Phys.} \textbf{\bibinfo{volume}{11}},
  \bibinfo{pages}{085006} (\bibinfo{year}{2009}).

\bibitem[{\citenamefont{Laurent and Massart}(2000)}]{LM00}
\bibinfo{author}{\bibfnamefont{B.}~\bibnamefont{Laurent}} \bibnamefont{and}
  \bibinfo{author}{\bibfnamefont{P.}~\bibnamefont{Massart}},
  \bibinfo{journal}{The Annals of Statistics} \textbf{\bibinfo{volume}{28}},
  \bibinfo{pages}{1302} (\bibinfo{year}{2000}).

\end{thebibliography}
\end{document}